\documentclass[draft]{birkau}
\usepackage{amsfonts}
\usepackage{amsmath}
\usepackage{amsthm}
\usepackage{latexsym}
\usepackage{epsfig}
\usepackage{latexsym}

\begin{document}
%%%%%%%%%%%%%%%%%%%%%%%%%%%%%%%%%
% Enviroments
\newtheorem{theorem}{Theorem}
\newtheorem{corollary}[theorem]{Corollary}
\newtheorem{prop}[theorem]{Proposition}
\newtheorem{problem}[theorem]{Problem}
\newtheorem{lemma}[theorem]{Lemma}
\newtheorem{remark}[theorem]{Remark}
\newtheorem{observation}[theorem]{Observation}
\newtheorem{defin}{Definition}
\newtheorem{example}{Example}
\newtheorem{conj}{Conjecture}

%%%%%%%%%%%%%%%%%%%%%%%%%%%%%%%%%%%
% Operators
\newcommand{\Con}{{\sf Con}}
\renewcommand{\H}{{\sf H}}
\renewcommand{\S}{{\sf S}}
\newcommand{\Term}{{\sf Term}}
\newcommand{\var}[1]{{\sf var}(#1)}
\newcommand{\Sg}[1]{{\sf Sg}(#1)}
\def\Sgg#1#2{{\sf Sg}_{#1}(#2)}
\newcommand{\tol}{{\sf tol}}
\newcommand{\lnk}{{\sf lk}}
%%%%%%%%%%%%%%%%%%%%%%%%%%%%%%%%%%
% Operations
\let\cd=\cdot
\let\eq=\equiv
\let\op=\oplus
\let\meet=\wedge
\let\join=\vee
\let\tm=\times
%%%%%%%%%%%%%%%%%%%%%%%%%%%%%%%%%%
% Tame congruence
\newcommand{\one}{{\bf1}}
\newcommand{\two}{{\bf2}}
%%%%%%%%%%%%%%%%%%%%%%%%%%%%%%%%%%%
% Accents and so
\let\un=\underline
\let\ov=\overline
%%%%%%%%%%%%%%%%%%%%%%%%%%%%%%%%%%%%%%%%%%
% Abbreviations for algebras and clones
\newcommand{\zA}{{\mathbb A}}
\newcommand{\zB}{{\mathbb B}}
\newcommand{\zC}{{\mathbb C}}
\newcommand{\zD}{{\mathbb D}}
\newcommand{\zM}{{\mathbb M}}
\newcommand{\cD}{{\mathcal D}}
\newcommand{\cE}{{\mathcal E}}
\newcommand{\cG}{{\mathcal G}}
\newcommand{\cK}{{\mathcal K}}
\newcommand{\cV}{{\mathcal V}}
%%%%%%%%%%%%%%%%%%%%%%%%%%%%%%%%%%%%%%%%
% Vectors
\newcommand{\ba}{{\bf a}}
\newcommand{\bb}{{\bf b}}
\newcommand{\bc}{{\bf c}}
\newcommand{\bd}{{\bf d}}
\newcommand{\be}{{\bf e}}
\newcommand{\oa}{\ov a}
\newcommand{\ob}{\ov b}
%%%%%%%%%%%%%%%%%%%%%%%%%%%%%%%%%%%%%%%%
% Constraint satisfaction Problem
\newcommand{\rel}{R}
\newcommand{\relo}{Q}
\newcommand{\rela}{S}
\newcommand{\dep}{\mathsf{dep}}
\newcommand{\Filt}{\mathrm{Ft}}
\newcommand{\amax}{\mathsf{amax}}
\newcommand{\umax}{\mathsf{umax}}
\newcommand{\as}{\mathsf{as}}
\newcommand{\asm}{\mathsf{asm}}
\newcommand{\se}[1]{\mathsf{s}(#1)}
\def\see#1#2{\mathsf{s}_{#1}(#2)}
\let\sqq=\sqsubseteq

%%%%%%%%%%%%%%%%%%%%%%%%%%%%%%%%%%%%%%%%
% Mathematical abbreviations
\let\sse=\subseteq
\def\vc#1#2{#1 _1\zd #1 _{#2}}
\def\tms#1#2{#1 _1\tm\dots\tm #1 _{#2}}
\newcommand{\zd}{,\ldots,}
\newcommand{\red}[1]{\vrule height7pt depth3pt width.4pt
\lower3pt\hbox{$\scriptstyle #1$}}
\newcommand{\fac}[1]{/\lower2pt\hbox{$\scriptstyle #1$}}
\newcommand{\eqc}[1]{\stackrel{#1}{\eq}}
\def\cl#1#2{\arraycolsep0pt
\left(\begin{array}{c} #1\\ #2 \end{array}\right)}
\def\cll#1#2#3{\arraycolsep0pt \left(\begin{array}{c} #1\\ #2\\
#3 \end{array}\right)}
\newcommand{\pr}{{\mathrm{pr}}}
\newcommand{\maj}{\mathrm{maj}}

%%%%%%%%%%%%%%%%%%%%%%%%%%%%%%%%%%%%%%%%%%%
% Other abbreviations
\newcommand{\lb}{$\linebreak$}
%%%%%%%%%%%%%%%%%%%%%%%%%%%%%%%%%%%%%%%%%%
% Greek symbols
\let\al=\alpha
\let\vf=\varphi
\let\th=\theta
\let\Dl=\Delta
\newcommand{\eps}{\emptyset}

%%%%%%%%%%%%%%%%%%%%%%%%%%%%%%%%%%%%%%
%%%%%%%%%%%%%%%%%%%%%%%%%%%%%%%%%%%%%%
\title[Graphs of finite algebras II]{Graphs of finite algebras: maximality, rectangularity, and decomposition}
\corrauthor[A. A. Bulatov]{Andrei A.\ Bulatov}
\address{School of Computing Science, Simon Fraser University, Burnaby, Canada}
\urladdr{www.cs.sfu.ca/~abulatov}
\email{abulatov@cs.sfu.ca}
\thanks{This work was supported by an NSERC Discovery grant. }

\begin{abstract}
In this paper we continue the study of edge-colored graphs associated with
finite idempotent algebras initiated in [A.Bulatov, ``Local structure of 
idempotent algebras I'', CoRR, abs/2006.09599, 2020.]. We prove stronger 
connectivity properties of such 
graphs that will allows us to demonstrate several useful structural features
of subdirect products of idempotent algebras such as rectangularity and 
2-decomposition. 
\end{abstract}

\maketitle

%%%%%%%%%%%%%%%%%%%%%%%%%%%%%%%%%
%%%%%%%%%%%%%%%%%%%%%%%%%%%%%%%%%
\section{Introduction and Preliminaries}\label{sec:intro}

We continue the study of graphs associated with algebras that was initiated 
in the first part of this paper \cite{Bulatov20:graph} (see also 
\cite{Bulatov16:graph}). 
The vertices of the graph $\cG(\zA)$ associated with an idempotent algebra
$\zA$ are the elements of $\zA$, and the edges are pairs of vertices
that have certain properties with respect to term operations of 
$\zA$, see the definitions below. Two kinds of edges were introduced, 
`thick' and `thin', where the thin version is a more technical kind, perhaps 
less intuitive, but also more suitable as a tool for the results of this part. 
Thin edges are also directed converting $\cG(\zA)$ into a digraph.
In the second part we first focus on the connectivity properties of this 
digraph, in particular, we show that more vertices are connected by directed paths 
of thin edges than one might expect (Theorem~\ref{the:connectivity}). 
Then we study the structure of subdirect products of algebras. 
An important role here is played by so-called as-components of algebras,
which are subsets of algebras defined through certain connectivity properties.
We prove that subdirect products when restricted to as-components 
have the property of rectangularity similar to Mal'tsev algebras, and also
similar to the Absorption Theorem (see, e.g.\ \cite{Barto12:absorbing}). 
Finally, we show that,
again, modulo as-components any subdirect product is 2-decomposable,
similar to Baker-Pixley Theorem \cite{Baker75:chinese-remainder}.

%%%%%%%%%%%%%%%%%%%%%%%%%%%%%%%%%
\subsection{Thick edges}

We start with recalling the main definitions and results from 
\cite{Bulatov20:graph}. Operation $x-y+z$ of a module $\zM$ is said 
to be \emph{affine}. Let $\zA$ be a finite algebra with universe 
$A$. Recall that for $B\sse A$ the subalgebra of $\zA$ generated by 
$B$ is denoted $\Sgg\zA B$, or just $\Sg B$ if $\zA$ is clear from the 
context. Edges of $\zA$ are defined as follows.
A pair $ab$ of vertices is an \emph{edge} if and only if
there exists a congruence $\th$ of $\Sg{a,b}$ such that either
$\Sg{a,b}\fac\th$ is a set, or it is term equivalent to the full idempotent
reduct of a module (we will simply say that $\Sg{a,b}\fac\th$ is a module) 
and there is a term operation of $\zA$ such
that $f\fac\th$ is an affine operation of $\Sg{a,b}\fac\th$, or
there exists a term operation of $\zA$ such that $f\fac\th$ is a semilattice 
operation on $\{a\fac\th,b\fac\th\}$, or $f\fac\th$
is a majority operation on $\{a\fac\th,b\fac\th\}$. 

If there exists a congruence and a term operation of $\zA$ such that
$f\fac\th$ is a semilattice operation on $\{a\fac\th,b\fac\th\}$ then $ab$ is
said to have the {\em semilattice} type. Edge $ab$ is of the
{\em majority} type if there are a congruence $\th$ and 
$f\in\Term(\zA)$ such that $f\fac\th$ is a majority
operation on $\{a\fac\th,b\fac\th\}$, but there is no term operation of $\zA$ which is
semilattice on this set. Pair $ab$ has the {\em affine type}
if there are a congruence $\th$ and $f\in\Term(\zA)$
such that $\Sg{a,b}\fac\th$ is a module and $f\fac\th$ is its 
affine operation. Finally, $ab$ is of the
\emph{unary type} if $\Sg{a,b}\fac\th$ is a set. In all cases we say that
congruence $\th$ \emph{witnesses} the type of edge $ab$. The set 
$\{a\fac\th, b\fac\th\}$ will often be referred to as a \emph{thick} edge.

In this paper we assume that $\zA$ does not have edges of the unary type,
which is equivalent to the statement that $\var\zA$, the variety generated by 
$\zA$, omits type \one\ (Theorem 5(2) of \cite{Bulatov20:graph}).
We restate the relevant results from \cite{Bulatov20:graph}.

% def of edges
\begin{theorem}[Theorem~5(2,3) of \cite{Bulatov20:graph}]%
\label{the:connectedness}
Let $\zA$ be an idempotent algebra $\zA$ such that 
$\var\zA$ omits type \one. Then
\begin{itemize}
\item[(1)]
$\zA$ contains no edges of the unary type, and any two elements of 
$\zA$ are connected by a sequence of edges of the semilattice, majority, 
and affine types;
\item[(2)]
$\var\zA$ omits type \two\ if and only if $\zA$
contains no edges of the unary and affine types.
\end{itemize}
\end{theorem}

% def of smooth algebras
Algebra $\zA$ is said to be \emph{smooth} if for every edge $ab$ of
the semilattice or majority type, $a\fac\th\cup b\fac\th$, where $\th$ is a 
congruence witnessing that $ab$ is an edge, is a subalgebra of $\zA$. 

\begin{theorem}[Theorem~12 of \cite{Bulatov20:graph}]%
\label{the:smooth}
For any idempotent algebra $\zA$ such that $\zA$ does not contain edges 
of the unary type there is a reduct $\zA'$ of $\zA$ that is smooth
and does not contain edges of the unary type. 

Moreover, if $\zA$ does not contain edges of the affine types, $\zA'$
can be chosen such that it does not contain edges of the affine type.
\end{theorem}

% def of thin edges
% unified op's
In the case of smooth algebras the operations involved in the 
definition of edges can be significantly unified.

\begin{theorem}[Theorem~21, Corollary~22 of \cite{Bulatov20:graph}]%
\label{the:uniform}
Let $\cK$ be a finite set of similar smooth idempotent algebras. 
There are term operations $f,g,h$ of $\cK$ such that for every
edge $ab$ of $\zA\in\cK$, where $\th$ is a congruence of $\Sg{a,b}$
witnessing that $ab$ is an edge and $B=\{a\fac\th,b\fac\th\}$
\begin{description}
\item[(i)]
$f\red B$ is a semilattice
operation if $ab$ is a semilattice edge, and it is the first projection if $ab$ is a
majority or affine edge;
\item[(ii)]
$g\red B$ is a majority operation if $ab$ is a majority edge, it is the first
projection if $ab$ is an affine edge, and $g\red B(x,y,z)=f\red B(x,f\red B(y,z))$ 
if $ab$ is semilattice;
\item[(iii)]
$h\red{\Sg{a,b}\fac\th}$ is an affine operation
operation if $ab$ is an affine edge, it is the first
projection if $ab$ is a majority edge, and
$h\red B(x,y,z)=f\red B(x,f\red B(y,z))$ if $ab$ is semilattice.
\end{description}
\end{theorem}

Operations $f,g,h$ from Theorem~\ref{the:uniform} above can be chosen 
to satisfy certain identities.

\begin{lemma}[Lemma~23 of \cite{Bulatov20:graph}]\label{lem:fgh-identities}
Operations $f,g,h$ identified in Theorem~\ref{the:uniform} can be chosen 
such that
\begin{itemize}\itemsep0pt
\item[(1)]
$f(x,f(x,y))=f(x,y)$ for all $x,y\in\zA\in\cK$;
\item[(2)]
$g(x,g(x,y,y),g(x,y,y))=g(x,y,y)$ for all $x,y\in\zA\in\cK$;
\item[(3)]
$h(h(x,y,y),y,y)=h(x,y,y)$ for all $x,y\in\zA\in\cK$.
\end{itemize}
\end{lemma}

\begin{prop}[Lemmas~8,10 of \cite{Bulatov20:graph}]\label{pro:edge-extend}
Let $\zA$ be an idempotent algebra. 
Then
\begin{itemize}
\item[(1)] 
if $ab$ is an edge in $\zA$, it is an edge of the same type in any 
subalgebra $\zB$ of $\zA$ containing $a,b$;
\item[(2)] 
if $\al$ is a congruence of $\zA$ and $a\fac\al b\fac\al$, $a,b\in\zA$,
 is an edge in $\cG(\zA\fac\al)$ then $ab$ is also an edge in $\cG(\zA)$ 
 of the same type.
 \end{itemize}
\end{prop}

Note that, as Example~9 from \cite{Bulatov20:graph} shows that
if $ab$ is an edge in $\zA$ and $\al\in\Con(\zA)$, then $a\fac\al b\fac\al$
does not have to be an edge even if $a\fac\al\ne b\fac\al$.

%%%%%%%%%%%%%%%%%%%%%%%%%%%%%%%%%
\subsection{Thin edges}

Thin edges, also introduced in \cite{Bulatov20:graph}, offer a better 
technical tool. Here we generalize this concept to algebras that are not 
necessarily smooth.

Let $\cK$ be a finite class of smooth idempotent algebras closed under 
taking subalgebras and homomorphic images and $\cV$ the class of finite 
algebras from the variety it generates, that is, the pseudovariety generated 
by $\cK$. We will slightly abuse the terminology and call $\cV$ the variety 
generated by $\cK$.  If we are interested in a particular algebra $\zA$, set 
$\cK=\H\S(\zA)$. Fix operations $f,g,h$ satisfying the conditions of 
Theorem~\ref{the:uniform} and Lemma~\ref{lem:fgh-identities}. 
For $\zA\in\cK$ and $a,b\in\zA$, the pair $ab$ is called a 
\emph{thin semilattice edge} if the equality relation
witnesses that it is a semilattice edge; or in other words if $f(a,b)=f(b,a)=b$.
The binary operation $f$ from Theorem~\ref{the:uniform} can be 
chosen to satisfy a special property. 

\begin{prop}[Proposition~24, \cite{Bulatov20:graph}]\label{pro:good-operation}
Let $\cK$ be a finite class of similar smooth idempotent algebras. There 
is a binary term operation $f$ of $\cK$ such that $f$ is a semilattice 
operation on every thick semilattice edge of every $\zA\in\cK$ and for 
any $a,b\in\zA$, $\zA\in\cK$, either $a=f(a,b)$ or the pair $(a,f(a,b))$ 
is a thin semilattice edge.
\end{prop}

We assume that operation $f$ satisfying the conditions of 
Proposition~\ref{pro:good-operation} is fixed, and use $\cdot$ to
denote it (think multiplication). If $ab$ is a thin semilattice edge,
that is, $a\cdot b=b\cdot a=b$, we write $a\le b$.

Let $\zA\in\cV$, $a,b\in\zA$, $\zB=\Sg{a,b}$, and $\th$ a congruence 
of $\zB$. Pair $ab$ is said to be 
\emph{minimal} with respect to $\th$ if for any $b'\in b\fac\th$, 
$b\in\Sg{a,b'}$. A ternary term $g'$ is said to satisfy the \emph{majority
condition} with respect to $\cK$ if it satisfies the conditions of 
Lemma~\ref{lem:fgh-identities}(2), and $g'$ is a majority operation on every 
thick majority edge of every algebra from $\cK$. A ternary term operation 
$h'$ is said to satisfy the \emph{minority condition} if it satisfies the conditions 
of Lemma~\ref{lem:fgh-identities}(3), and $h'$ is a Mal'tsev 
operation on every thick affine edge of every algebra from $\cK$. By 
Theorem~\ref{the:uniform} and Lemma~\ref{lem:fgh-identities} operations 
satisfying the majority and minority conditions exist.

Let $\zA\in\cV$ and $a,b\in\zA$. The pair $ab$ is a \emph{thin semilattice edge} 
if the term $\cdot$ of $\cV$ is a semilattice operation on $\{a,b\}$ and $ab=b$. It is 
said to be a \emph{thin majority edge} if 
\begin{itemize}
\item[(*)] 
for any term operation $g'$ satisfying the majority condition with respect to $\cK$,
the subalgebras $\Sg{a,g'(a,b,b)},\Sg{a,g'(b,a,b)},\Sg{a,g'(b,b,a)}$
contain~$b$.
\end{itemize}
If in addition to the condition above $ab$ is also a majority edge, a congruence
$\th$ witnesses that, and $ab$ is a minimal pair with respect to $\th$, we 
say that $ab$ is a \emph{special majority edge}.
The pair $ab$ is called a \emph{thin affine edge} if for any term operation $h'$ 
satisfying the minority condition \emph{with respect to $\cK$}
\begin{itemize}
\item[(**)] 
$h(b,a,a)=b$ and $b\in\Sg{a,h'(a,a,b)}$.
\end{itemize}

The operations $g,h$ from Theorem~\ref{the:uniform} do not 
have to satisfy any specific conditions on the set $\{a,b\}$, when
$ab$ is a thin majority or affine edge, except what follows from their 
definition. Also, both thin majority and thin affine edges are 
directed, since $a,b$ in the definition occur asymmetrically. Note also, that 
which pairs of an algebra $\zA$ are thin majority and affine 
edges depends not only on the algebra itself, but also on the underlying
class $\cK$ and the choice of operations $f,g,h$. 

It was shown in \cite{Bulatov20:graph} that in smooth algebras every 
edge has a thin edge associated with it. 

\begin{lemma}[Corollaries~25,29,33 Lemmas~28,32, \cite{Bulatov20:graph}]%
\label{lem:thin-semilattice}
Let $\zA\in\cK$ and let $ab$ be a semilattice (majority, affine) edge, 
$\th$ a congruence of
$\Sg{a,b}$ that witnesses this, and $c\in a\fac\th$. If $ab$ is 
a semilattice or majority edge, then for any $d\in b\fac\th$ such that 
$cd$ is a minimal pair with respect to $\th$ the pair $cd$ is a thin 
semilattice or special majority edge. If $ab$ is affine then for any 
$d\in b\fac\th$ such that $ad$ is a minimal pair with respect to 
$\th$ and $h(d,a,a)=d$ the pair $ad$ is a thin affine edge. Moreover,
$d\in b\fac\th$ satisfying these conditions exists.
\end{lemma}

First of all we observe that Proposition~\ref{pro:good-operation} generalizes to 
algebras from $\cV$. To prove Proposition~\ref{pro:very-good-operation} 
below it suffices
to observe that the property of the operation $\cdot$ stated in 
Proposition~\ref{pro:good-operation} can be expressed as identities
\[
x\cdot(x\cdot y)=(x\cdot y)\cdot x=x\cdot y.
\]

\begin{prop}\label{pro:very-good-operation}
For every $\zA\in\cV$ and for any $a,b\in\zA$ either $a=a\cdot b$ or 
the pair $(a,a\cdot b)$ is a thin semilattice edge.
\end{prop}

We will need statements similar to Lemmas~31, 35 from \cite{Bulatov20:graph} for 
algebras in $\cV$. 

\begin{lemma}\label{lem:thin-combination-variety}
\begin{itemize}
\item[(1)] 
Let $ab$ and $cd$ be thin affine edges in $\zA_1,\zA_2\in\cV$. Then 
there is an operation $h'$ such that $h'(b,a,a)=b$ and 
$h'(c,c,d)=d$. In particular, for any thin affine edge $ab$ there 
is an operation $h'$ such that $h'(b,a,a)=h'(a,a,b)=b$.
\item[(2)] 
Let $ab$ and $cd$ be thin edges in $\zA_1,\zA_2\in\cV$.
If they have different types there is a binary term operation $p$ such that
$p(b,a)=b$, $p(c,d)=d$. 
\end{itemize}
\end{lemma}

\begin{proof}
(1) Let $\rel$ be the subalgebra of $\zA_1\tm\zA_2$ generated by 
$(b,c),(a,c),(a,d)$. By the definition of thin affine edges,
\[
\cl b{d'}=h\left(\cl bc,\cl ac,\cl ad\right)\in\rel,
\]
where $h$ is the operation fixed in the beginning of the section.
Then as $d\in\Sgg{\zA_2}{c,h(c,c,d)}$, there is a term operation 
$r(x,y)$ such that $d=r(c,d')$. Therefore 
\[
\cl bd=r\left(\cl bc,\cl b{d'}\right)\in\rel.
\] 
The result follows.

(2) Let $\rel$ be the subalgebra of $\zA_1\tm\zA_2$ generated by 
$(b,c),(a,d)$. 

If $ab$ is majority and $c\le d$, let $g'(x,y,z)=g(x,y\cdot x, z\cdot x)$.
Since $x\cdot y=x$ on every (thick) majority edge of every algebra from 
$\cK$, the operation $g'$ is a majority operation on every thick majority edge
of any algebra from $\cK$.   We use the construction from the 
proof of Lemma~\ref{lem:fgh-identities}(2) (Lemma~16(2) of 
\cite{Bulatov20:graph}). Consider the
unary operation $g_x(y)=g'(x,y,y)$. Clearly, for some $n$, the operation 
$g_x^n$ is idempotent  for every $\zB\in\cK$ and $x\in\zB$. Set 
$g''(x,y,z)=g_x^{n-1}(g'(x,y,z))$. Then it is not hard to see that $g''$ is a 
majority operation on every thick majority edge and that 
$g''(x,g''(x,y,y), g''(x,y,y))=g''(x,y,y)$ is an identity in $\cV$ (see the 
proof of Lemma~16(2) of \cite{Bulatov20:graph} for details). Also, as 
$g'(d,c,c)=g(d,cd,cd)=d$, we have $g''(d,c,c)=d$.

Therefore the operation $g''$ satisfies the majority condition. 
By the definition of thin majority edges there is a binary term operation $r$ 
such that $b=r(a,g''(a,b,b))$. Then
\[
\cl bd=
r\left(\cl ad,\cl{g''(a,b,b)}d\right)=
r\left(\cl ad,g''\left(\cl ad,\cl bc,\cl bc\right)\right)\in\rel.
\]
Therefore $p(x,y)=r(y,g''(y,x,x))$ satisfies the conditions required.

If $ab$ is affine and $c\le d$, let $h'(x,y,z)=h(x,y\cdot x,z\cdot x)$.
Since $x\cdot y=x$ on every (thick) affine edge of every algebra from 
$\cK$, the operation $h'$ is a Mal'tsev operation on every thick affine 
edge of every $\zB\in\cK$. We use the construction from the 
proof of Lemma~\ref{lem:fgh-identities}(3) (Lemma~23(3) of 
\cite{Bulatov20:graph}). Let $h_y(x)=h'(x,y,y)$. Clearly, for some $n$, 
the operation $h_y^n$ is idempotent  for every $\zB\in\cK$ and $y\in\zB$. 
Let $h'_0(x,y,z)=h'(x,y,z)$ and 
$h'_{i+1}(x,y,z)=h'_i(h(x,y,y),y,z)$ for $i\ge0$. Then $h'_i(x,y,y)=h_y^i(x)$. 
Hence $h'_n(h'_n(x,y,y),y,y)=h'_n(x,y,y)$ is an identity in $\cV$. It is not hard 
to see that $h''(x,y,z)=h'_n(x,y,z)$ is a Mal'tsev operation on every thick affine 
edge of every $\zB\in\cK$. Also, as $h'(d,d,c)=h(d,dd,cd)=d$, we have
$h''(d,d,c)=d$.

Thus, $h''$ satisfies the minority condition. Hence
by the definition of thin affine edges there is a binary term operation $r$ 
such that $b=r(a,h''(a,a,b))$. Then
\[
\cl bd=
r\left(\cl ad,\cl{h''(a,a,b)}d\right)
=r\left(\cl ad,h''\left(\cl ad,\cl ad,\cl bc\right)\right)\in\rel.
\] 
Therefore $p(x,y)=r(y,h''(y,y,x))$ satisfies the conditions required.

If $ab$ is affine and $cd$ is majority, then set 
\[
\cl {b'}{d'}=h\left(\cl ad,\cl ad,\cl bc\right)\in\rel.
\]
By the definition of thin affine edges there is a binary term operation $r$ 
such that $b=r(a,h(a,a,b))$. Consider the operation
\[
g'(z,y,x)=g(r(z,h(z,y,x)),r(y,h(y,z,x)),x).
\]
It satisfies the following condition.
\begin{align*}
g'(b,a,a) & =g(r(a,h(a,a,b)),r(a,h(a,a,b)),b)\\
& =g(r(a,b'),r(a,b'),b)=g(b,b,b)=b.
\end{align*}
Also, on any thick majority edge $\{c'\fac\th,d'\fac\th\}$ of an algebra $\zB\in\cK$, 
where $\th$ witnesses that $c'd'$ is a majority edge, $h(x,y,z)=x$, therefore
\begin{align*}
g'(x,y,z) & =g(r(z,h(z,y,x)),r(y,h(y,z,x)),x)\\
& =g(r(z,z),r(y,y),x)=g(z,y,x)
\end{align*}
on $\Sgg\zB{c',d'}\fac\th$. 

Next, we use the argument from the beginning of item (2). 
Consider the unary operation $g_x(y)=g'(x,y,y)$. Let $n$ be such that 
the operation $g_y^n$ is idempotent  for every $\zB\in\cK$ and $y\in\zB$, and set 
$g''(x,y,z)=g_x^{n-1}(g'(x,y,z))$. Then $g''$ is a 
majority operation on every thick majority edge of any $\zB\in\cK$, and  
$g''(x,g''(x,y,y), g''(x,y,y))=g''(x,y,y)$ is an identity in $\cV$. Also, as 
$g'(b,a,a)=b=g_b(a)=g_b(b)$, we have $g_b^n(a)=b$, and therefore $g''(b,a,a)=b$.

Therefore the operation $g''$ satisfies the majority condition, and since $cd$ is 
a thin majority edge, there exists a binary term operation $s$ such that 
$s(c,g''(c,d,d))=d$.
Thus,
\[
\cl bd=s\left(\cl bc,\cl b{g''(c,d,d)}\right)=
s\left(\cl bc,g''\left(\cl bc,\cl ad,\cl ad\right)\right)\in\rel.
\]
The result follows.
\end{proof}

%%%%%%%%%%%%%%%%%%%%%%%%%%%%%%%
\subsection{Paths and filters in smooth algebras}

Let $\zA\in\cK$ be a smooth algebra. A \emph{path} in $\zA$ is a sequence 
$a_0,a_1\zd a_k$ such 
that $a_{i-1}a_i$ is a thin edge for all $i\in[k]$ (note that thin edges 
are always assumed to be directed).
We will distinguish paths of several types depending on what types of
edges are allowed. If $a_{i-1}\le a_i$ for all $i\in[k]$ then the path is
called a \emph{semilattice} or \emph{s-path}. If for every $i\in[k]$
either $a_{i-1}\le a_i$ or $a_{i-1}a_i$ is a thin
affine edge then the path is called \emph{affine-semilattice} or
\emph{as-path}. 
The path is called \emph{asm-path} when all types of edges are allowed. If 
there is a path $a=a_0,a_1\zd a_k=b$ which is arbitrary (semilattice, 
affine-semilattice) then $a$ is said to be
\emph{asm-connected} (or \emph{s-connected}, or 
\emph{as-connected}) to $b$. We will also 
say that $a$ is \emph{connected} to $b$ if it is asm-connected. 
We denote this by $a\sqq^{asm}b$
(for asm-connectivity), $a\sqq b$, and $a\sqq^{as}b$ for s-, and as-connectivity, respectively. 
If all the thin majority edges in an asm-path are special,
we call such path special. The following is a direct implication of 
Theorem~\ref{the:connectedness} and Lemma~\ref{lem:thin-semilattice}.

\begin{corollary}\label{cor:oriented}
Any two elements of a smooth algebra $\zA\in\cK$ are connected by an 
oriented path consisting of thin edges (i.e.\ a path in which some edges can be traversed backwards).
\end{corollary}

Let $\cG_s(\zA)$ ($\cG_{as}(\zA),\cG_{asm}(\zA)$) denote the digraph 
whose nodes are the elements of $\zA$, and the edges are the thin 
semilattice edges (thin semilattice and affine edges, arbitrary thin edges, 
respectively). The strongly connected component of $\cG_s(\zA)$ 
containing $a\in\zA$ will be denoted by $\se a$. The set of strongly 
connected components of $\cG_s(\zA)$ are ordered in the natural 
way (if $a\le b$ then $\se a\le \se b$), the elements belonging to 
maximal ones will be called \emph{maximal}, and the set of all 
maximal elements from $\zA$ will be denoted by $\max(\zA)$. 

The strongly connected component of $\cG_{as}(\zA)$ containing 
$a\in\zA$ will be denoted by $\as(a)$. A maximal strongly connected 
component of this graph is called an
\emph{as-component}, an element from an as-component
is called \emph{as-maximal}, and the set of all as-maximal elements is
denoted by $\amax(\zA)$. 

Finally, the strongly connected component of $\cG_{asm}(\zA)$ containing 
$a\in\zA$ will be denoted by $\asm(a)$. A maximal strongly connected 
component of this graph is called an \emph{universally maximal 
component} (or \emph{u-maximal component} for short), an element 
from a u-component is called \emph{u-maximal}, and the set of all 
u-maximal elements is denoted by $\umax(\zA)$. 

Alternatively, maximal, as-maximal, and u-maximal elements 
can be characterized as follows: an element $a\in\zA$ is 
maximal (as-maximal, u-maximal) if for every $b\in\zA$ such 
that $a\sqq b$ ($a\sqq^{as}b$, $a\sqq^{asm} b$)
it also holds that $b\sqq a$ ($b\sqq^{as}a$, $b\sqq^{asm}a$). 
Sometimes it will be convenient to specify what the algebra is, 
in which we consider maximal components, as-components, 
or u-maximal components, and the corresponding connectivity. 
In such cases we will specify it by writing $\see \zA a$, 
$\as_\zA(a)$, or $\asm_\zA(a)$. For connectivity we will use
$a\sqq_\zA b$, $a\sqq_\zA^{as}b$, and $a\sqq_\zA^{asm}b$.

By $\Filt_\zA(a)=\{b\in\zA\mid a\sqq_\zA b\}$ we
denote the set of elements $a$ is connected to (in terms of 
semilattice paths); similarly, by 
$\Filt^{as}_\zA(a)=\{b\in\zA\mid a\sqq_\zA^{as} b\}$ and
$\Filt^{asm}_\zA(a)=\{b\in\zA\mid a\sqq_\zA^{asm} b\}$ we 
denote the set of elements $a$ is as-connected and 
asm-connected to. Also, $\Filt_\zA(C)=\bigcup_{a\in C}\Filt_\zA(a)$
($\Filt^{as}_\zA(C)=\bigcup_{a\in C}\Filt^{as}_\zA(a)$,
$\Filt^{asm}_\zA(C)=\bigcup_{a\in C}\Filt^{asm}_\zA(a)$, 
respectively) for $C\sse \zA$. Note that if $a$ is a maximal 
(as-maximal or u-maximal) element then $\se a=\Filt_\zA(a)$
($\as(a)=\Filt^{as}_\zA(a)$, and $\umax(a)\sse\Filt^{asm}_\zA(b)$).

%%%%%%%%%%%%%%%%%%%%%%%%%%%%%
\subsection{Paths in non-smooth algebras}

Given the notion of thin edges, (s-, as-, asm-) paths, (as-, asm-) maximal elements in 
algebras from $\cV$ can be defined in the same way as for algebras 
in $\cK$, even if those algebras are not smooth. We use the same notation 
$\le$, $\sqq$, $\sqq^{as}$, $\sqq^{asm}$, $\se{a}$, $\as(a)$, $\asm(a)$, $\max(\zA)$, 
$\amax(\zA)$, $\umax(\zA)$ as before. In this section 
we study properties of such paths and 
maximal elements and the connections between 
(as-, asm-) maximal elements of an algebra with those in a
quotient algebra or subdirect product.

\begin{lemma}\label{lem:quotient-edge}
Let $\zA\in\cV$ and $\th\in\Con(\zA)$.
\begin{itemize}
\item[(1)] 
If $\oa\ob$ is a thin edge in $\zA\fac\th$ and $a\in\oa$, 
then there is $b\in\ob$ such that $ab$ is a thin edge in $\zA$
of the same type.
Morever, if $\zA\in\cK$ and $\oa\ob$ is a special thin majority edge, then 
so is $ab$. 
\item[(2)] 
If $ab$ is a thin edge in $\zA$, then $a\fac\th b\fac\th$
is a thin edge in $\zA\fac\th$.
\end{itemize}
\end{lemma}  

\begin{remark}
Note that in Lemma~\ref{lem:quotient-edge}(2) if $ab$
is a special thin majority edge, there is no guarantee 
that $a\fac\th b\fac\th$ is also a special edge.
\end{remark}

\begin{proof}
Pick an arbitrary $b\in\ob$ such that the pair
$ab$ is minimal with respect to $\th$. Let 
$\zB=\Sgg\zA{a,b}$. 

First, if $\oa\ob$ is a thin semilattice edge, then 
$a\cdot b\in\oa\cdot\ob=\ob$. As $a\ne b$, by 
Proposition~\ref{pro:very-good-operation} $(a,ab)$ is a semilattice edge.  

Now suppose  that $\oa\ob$ is a thin majority edge. 
Let $g'$ be an operation satisfying the majority property 
with respect to $\cK$. Let $c=g'(a,b,b)$. Since $\oa\ob$ 
is a thin majority edge, there is a binary term operation $t$
such that $t(\oa,c\fac\th)=\ob$. Note that 
$t(a,c)\in\zB$, as well. By the choice of $b$, it holds 
that $b\in\Sg{a,t(a,c)}$. Therefore, $b\in\Sg{a,g'(a,b,b)}$. 
That $b\in\Sg{a,g'(b,a,b)}$ and $b\in\Sg{a,g'(b,b,a)}$ 
can be proved in the same way. Suppose in addition that $\zA\in\cK$ 
and $\oa\ob$ is a special thin majority edge, that is, it is 
a majority edge in $\zA\fac\th$ and it is witnessed by a 
congruence $\eta$ of $\Sgg{\zA\fac\th}{\oa,\ob}$. 
It is then straightforward that $ab$ is a majority edge 
as witnessed by the congruence 
$\eta\fac\th=\{(d,e)\in\zB^2\mid d\fac\th\eqc\eta e\fac\th\}$.

Finally, suppose that $\oa\ob$ is a thin affine edge, and set
$b'=h(b,a,a)$, where $h$ is the operation identified in 
Theorem~\ref{the:uniform}(iii). By Lemma~\ref{lem:fgh-identities}(3)
it holds that $h(b',a,a)=b'$. Note that $ab'$ is a minimal pair as well.
Consider $c=h'(a,a,b')$, for any $h'$ satisfying the minority condition 
with respect to $K$. Since $\oa\ob$
is a thin affine edge, $\ob\in\Sgg{\zA\fac\th}{\oa,c\fac\th}$.
This means that there is $d\in\Sgg\zA{a,c}$ such that 
$d\in\ob$. Since $d\in\Sgg\zA{a,b'}$ and $ab'$ is a minimal 
pair with respect to $\th$, we obtain $b'\in\Sgg\zA{a,d}
\sse\Sgg\zA{a,c}$. Thus $ab'$ is a thin affine edge.

\smallskip

(2) If $a\eqc\th b$, the statement of the lemma is trivial.
Otherwise if $a\le b$, then 
$a\fac\th\cdot b\fac\th=b\fac\th\cdot a\fac\th=b\fac\th$ showing that 
$a\fac\th\le b\fac\th$.
If $ab$ is a thin majority edge, then consider an abitrary
$g'$ satisfying the majority condition. Since 
$b\in\Sgg\zA{a,g'(a,b,b)}$, we also have 
$b\fac\th\in\Sgg{\zA\fac\th}{a\fac\th,g'(a\fac\th,b\fac\th,b\fac\th)}$.
The argument for the rest of condition (*) is similar.
If $ab$ is a thin affine edge, then $h(b,a,a)=b$, implying
$h(b\fac\th,a\fac\th,a\fac\th)=b\fac\th$. Also, for any $h'$ satisfying the minority condition, as 
$b\in\Sgg\zA{a,h'(a,a,b)}$, we have 
$b\fac\th\in\Sgg\zA{a\fac\th,h'(a\fac\th,a\fac\th,b\fac\th)}$.
\end{proof}

The next statement straightforwardly follows from 
Lemma~\ref{lem:quotient-edge}. 

\begin{corollary}\label{cor:quotient-path}
Let $\zA\in\cV$ and $\th\in\Con\zA$.
\begin{itemize}
\item[(1)] 
If $\vc{\oa}k$ is an s- (as-,asm-) path in $\zA\fac\th$ and 
$a\in \oa_1$, then there are $a_i\in\oa_i$ such that
$a_1=a$ and $\vc ak$ is an s- (as-, asm-) path in $\zA$.
Morever, if $\zA\in\cK$ and $\vc{\oa}k$ is a special asm-path, then 
$\vc ak$ is also a special asm-path. 
\item[(2)] If $\vc ak$ is an s- (as-, asm-) path in $\zA$, then 
$a_1\fac\th\zd a_\ell\fac\th$ is an s- (as-, asm-) 
path in $\zA\fac\th$.
\end{itemize}
\end{corollary}

\begin{corollary}\label{cor:quotient-max}
Let $\zA\in\cV$ and $\th\in\Con\zA$.
\begin{itemize}
\item[(1)] If $\oa\in\zA\fac\th$ is maximal (as-maximal, u-maximal) in 
$\zA\fac\th$, then there is $a\in\oa$ that is maximal (as-maximal, 
u-maximal) in $\zA$.
\item[(2)] If $a$ is maximal (as-maximal, u-maximal) in $\zA$, then $a\fac\th$ is 
maximal (as-maximal, u-maximal) in $\zA\fac\th$.
\end{itemize}
\end{corollary}

\begin{proof}
(1) We will use notation $c\sqq^x d$, $x\in\{s,as,asm\}$, to denote that
there is an $x$-path from $c$ to $d$. Pick an arbitrary $a'\in\oa$ and 
let $b$ be any $x$-maximal element of $\zA$ such that $a'\sqq^x b$.
This means that there is an $x$-path from $a'$ to $b$, and by
Corollary~\ref{cor:quotient-path}(2) there is also an $x$-path from 
$a'\fac\th$ to $b\fac\th$ in $\zA\fac\th$. Since by the assumption $\oa$ 
is $x$-maximal, there is also an $x$-path in $\zA\fac\th$ from
$b\fac\th$ to $\oa$. By Corollary~\ref{cor:quotient-path}(1)
there is an $x$-path in $\zA$ from $b$ to some element $a''\in\oa$.
Since $b$ is an $x$-maximal element, so is $a''$.

\smallskip

(2) Suppose that $a\fac\th\sqq^x b\fac\th$ for some $b\in\zA$. It suffices
to show that in this case $b\fac\th\sqq^x a\fac\th$. By 
Corollary~\ref{cor:quotient-path}(1) there is an $x$-path in $\zA$
from $a$ to some $b'\in b\fac\th$. Since $a$ is $x$-maximal, there also
exists an $x$-path $b'=b_1\zd b_k=a$ from $b'$ to $a$. By 
Corollary~\ref{cor:quotient-path}(2) 
$b\fac\th=b_1\fac\th\zd b_k\fac\th=a\fac\th$ 
is an $x$-path in $\zA\fac\th$ from $b\fac\th$ to $a\fac\th$.
\end{proof}

Next we study connectivity in subalgebras of direct products.

\begin{lemma}\label{lem:product-edge}
Let $\rel$ be a subalgebra of $\tms\zA n$, $\vc\zA n\in\cV$, 
$I\sse[n]$.
\begin{itemize}
\item[(1)] 
For any $\ba\in\rel$, $\ba^*=\pr_I\ba$, $\bb\in\pr_I\rel$ 
such that $\ba^*\bb$ is a thin edge, there is 
$\bb'\in\rel$, $\pr_I\bb'=\bb$, such that $\ba\bb'$ is a thin edge of the same 
type. 
\item[(2)] 
If $\ba\bb$ is a thin edge in $\rel$ then $\pr_I\ba\,\pr_I\bb$
is a thin edge in $\pr_I\rel$ of the same type (including the 
possibility that $\pr_I\ba=\pr_I\bb$).
\end{itemize}
\end{lemma}

\begin{proof}
Observe that we can consider $\pr_I\rel$ as the quotient algebra
$\rel\fac{\eta_I}$, where $\eta_I$ is the projection congruence
of $\rel$, that is, $(\bc,\bd)\in\eta_I$ if and only if 
$\pr_I\bc=\pr_I\bd$. Then item (1) can be rephrased as follows:
For any $\ba\in\rel$ and $\ob\in\rel\fac{\eta_I}$ such that
$\ba\fac{\eta_I}\ob$ is a thin edge in $\rel\fac{\eta_I}$, there is 
$\bb'\in\ob$ such that $\ba\bb'$ is a thin edge of the same 
type. It clearly follows from Lemma~\ref{lem:quotient-edge}(1).
Similarly, item (2) can be rephrased as: If $\ba\bb$ is a thin
edge in $\rel$, then $\ba\fac{\eta_I}\bb\fac{\eta_I}$ is a thin 
edge in $\rel\fac{\eta_I}$ of the same type. It follows 
straightforwardly from Lemma~\ref{lem:quotient-edge}(2).
\end{proof}

The next statement follows from Lemma~\ref{lem:product-edge}.

\begin{corollary}\label{cor:product-path}
Let $\rel$ be a subalgebra of $\tms\zA n$, $\vc\zA n\in\cV$, and 
$I\sse[n]$.
\begin{itemize}
\item[(1)] 
For any $\ba\in\rel$, and an s- (as-, asm-) path 
$\vc\bb k\in\pr_I\rel$ with $\pr_I\ba=\bb_1$, there is an
s- (as-, asm-) path $\vc{\bb'}k\in\rel$ such that 
$\bb'_1=\ba$ and $\pr_I\bb'_i=\bb_i$, $i\in[k]$. Moreover, if $\vc\bb k$ is a
special asm-path, so is $\vc{\bb'}k$.
\item[(2)] If $\vc\ba k$ is an s- (as-, asm-) path in $\rel$, then 
$\pr_I\ba_1\zd\pr_I\ba_k$ is an s- (as-, asm-) path in $\pr_I\rel$.
\end{itemize}
\end{corollary}

There is a connection between maximal (as-maximal, 
u-maximal) elements of a subdirect product and its projections.

\begin{corollary}\label{cor:product-maximal}
Let $\rel$ be a subalgebra of $\tms\zA n$, $\vc\zA n\in\cV$, 
and $I\sse[n]$.
\begin{itemize}
\item[(1)] For any maximal (as-maximal, u-maximal) (in $\pr_I\rel$) element 
$\bb\in\pr_I\rel$, there is $\bb'\in\rel$ which is maximal (as-maximal, 
u-maximal) in $\rel$ and such that $\pr_I\bb'=\bb$. In particular, 
$\pr_{[n]-I}\bb'$ is a maximal (as-maximal, u-maximal) in 
$\pr_{[n]-I}\rel$.
\item[(2)] If $\ba$ is maximal (as-maximal, u-maximal) in $\rel$, then 
$\pr_I\ba$ is maximal (as-maximal, u-maximal) in $\pr_I\rel$.
\end{itemize}
\end{corollary}

\begin{proof}
(1) We again can use the isomorphism between $\pr_I\rel$ and
the quotient algebra $\rel\fac{\eta_I}$. Then the statement of the
lemma translates into: If $\oa\in\rel\fac{\eta_I}$ is maximal (as-maximal, 
u-maximal) in $\rel\fac{\eta_I}$, then there is $a\in\oa$ that is maximal 
(as-maximal, u-maximal) in $\rel$. The latter statement is true by 
the second part of this corollary.

\smallskip

(2) Using the isomorphism between $\pr_I\rel$ and
the quotient algebra $\rel\fac{\eta_I}$, the statement follows from 
Corollary~\ref{cor:quotient-max}(2). Indeed, as $\ba$ is maximal 
(as-maximal, u-maximal)  in $\rel$, the element $\ba\fac{\eta_I}$ that 
corresponds to $\pr_I\ba$ is maximal (as-maximal, u-maximal) 
in $\rel\fac{\eta_I}$.
\end{proof}

The following lemma considers a special case of maximal components 
of various types in subdirect products.

\begin{lemma}\label{lem:as-product}
Let $\rel$ be a subdirect product of $\zA_1\tm\zA_2$, $B,C$ 
maximal components (as-components, u-components) of 
$\zA_1,\zA_2\in\cV$, respectively, and $B\tm C\sse\rel$. 
Then $B\tm C$ is a maximal component (as-component, u-component) 
of $\rel$.
\end{lemma}

\begin{proof}
We need to show that for any $\ba,\bb\in B\tm C$ there is an
s-path (as-path, asm-path) from $\ba$ to $\bb$. By the assumption
there is a (s-,as-,asm-) path $\ba[1]=a_1,a_2\zd a_k=\bb[1]$ from
$\ba[1]$ to $\bb[1]$ and a (s-,as-,asm-) path 
$\ba[2]=b_1,b_2\zd b_\ell=\bb[2]$ from $\ba[2]$ to $\bb[2]$.
Then the sequence 
$\ba=(a_1,b_1),(a_2,b_1)\zd(a_k,b_1)$, $(a_k,b_2)\zd(a_k,b_\ell)=\bb$
is a (s-,as,asm-) path from $\ba$ to $\bb$. Thus, $B\tm C$ is a connected component. 
It remains to observe that by Corollary~\ref{cor:product-maximal} it contains 
a maximal (as-maximal, u-maximal) element, and therefore is a maximal 
component (as-component, u-component) of $\rel$.
\end{proof}

%%%%%%%%%%%%%%%%%%%%%%%%%%%%%%%%%
%%%%%%%%%%%%%%%%%%%%%%%%%%%%%%%%%
\section{Connectivity}\label{sec:connectivity}

Recall that $\cK$ is a finite class of finite smooth idempotent algebras omitting type \one, and $\cV$ is the variety it generates.

%%%%%%%%%%%%%%%%%%%%%%%%%%%%%%%%
\subsection{General connectivity}

The main result of this section is that all maximal elements are connected to
each other. The undirected connectivity easily follows from the definitions,
and Theorem~\ref{the:connectedness}, so the challenge is to prove 
directed connectivity, as defined above. We start with an auxiliary lemma.

Let $\rel\le\zA_1\tm\dots\tm\zA_k$ be a relation. Also, let 
$\tol_i(\rel)$ (or simply $\tol_i$ if $\rel$ is clear from the context),
$i\in[k]$, denote the \emph{link} tolerance
\begin{align*}
& \{(a_i,a'_i)\in\zA_i^2\mid (a_1\zd a_{i-1},a_i,a_{i+1}\zd a_k),\\
& \quad
(a_1\zd a_{i-1},a'_i,a_{i+1}\zd a_k)\in\rel, \text{ for some
$(a_1\zd a_{i-1},a_{i+1}\zd a_k)$}\}.
\end{align*}
Recall that a tolerance is said to be \emph{connected} if its transitive 
closure is the full relation. The transitive closure 
$\lnk_i(\rel)$ of $\tol_i(\rel)$, $i\in[k]$, is called the 
\emph{link congruence}, and it is, indeed, a congruence. 

If $\rel$ is binary, that is, a subdirect product of $\zA_1,\zA_2$, 
then by $\rel[c], \rel^{-1}[c']$ for $c\in\zA_1,c'\in\zA_2$ 
we denote the sets $\{b\mid (c,b)\in\rel\}, \{a\mid (a,c')\in\rel\}$, 
respectively, and for $C\sse\zA_1,C'\sse\zA_2$ we use 
$\rel[C]=\bigcup_{c\in C}\rel[c]$, 
$\rel^{-1}[C']=\bigcup_{c'\in C'}\rel^{-1}[c']$, respectively. 
Relation $\rel$ is said to be \emph{linked} 
if the link congruences $\lnk_1(\rel),\lnk_2(\rel)$ are full 
congruences.

\begin{lemma}\label{lem:going-maximal}
Let $\zA\in\cV$ and let $\rela$ be a 
tolerance of $\zA$. Suppose that $(a,b)$, $a,b\in\zA$, 
belongs to the transitive closure of $\rela$, that is, 
there are $\vc d{k-1}$ such that $(d_i,d_{i+1})\in\rela$ for 
$i\in\{0,1\zd k-1\}$, where $d_0=a,d_k=b$. If for some 
$i\in\{0,1\zd k\}$ there is $d'_i\in\zA$ such that $d_i\sqq d'_i$, 
then there are $d'_j\in\zA$ for $j\in\{0,1,\zd k\}-\{i\}$ 
such that $d_j\sqq d'_j$ for $j\in\{0,1\zd k\}$, 
and $(d'_j,d'_{j+1})\in\rela$ for $j\in\{0,1\zd k-1\}$. 

Moreover, if $d_0\zd d_{i-1}$ are maximal, there are 
$d''_0\zd d''_k$ such that $d''_j=d_j$ for $j\in\{0\zd i-1\}$,
$d_i\sqq d'_i\sqq d''_i$, $d_j\sqq d''_j$ for $j\in\{i+1\zd k\}$, 
and $(d''_j,d''_{j+1})\in\rela$ for $j\in\{0\zd k-1\}$.
\end{lemma}

\begin{proof}
Let $d_i=d_i^1\le\ldots\le d_i^s=d'_i$ be an s-path from 
$d_i$ to $d'_i$. For each $j\in\{0,1\zd k\}-\{i\}$, we construct a 
sequence $d_j=d_j^1\le\ldots\le d_j^s$ by setting
\[
d_j^q=d_j^{q-1}\cdot d_i^q.
\]
Now, we prove by induction that $(d_j^q,d_{j+1}^q)\in\rela$ for all
$q\in[s]$ and $j\in\{0,1\zd k-1\}$. For $q=1$ it follows from the 
assumptions of the lemma. If $(d_j^q,d_{j+1}^q)\in\rela$, then 
$(d_j^{q+1},d_{j+1}^{q+1})=
(d_j^q\cdot d_i^{q+1},d_{j+1}^q\cdot d_i^{q+1})\in\rela$, since
$\rela$ is a tolerance. The first part of the lemma is proved.

For the second statement we apply the same construction, but only to 
elements $d_{i-1}\zd d_k$, and do not change $d_0\zd d_{i-2}$, that 
is, we set $d'_j=d_j$ for $j\in\{0\zd i-2\}$. Then, as before, we have 
$(d'_j,d'_{j+1})\in\rela$ for $j\in\{i-1\zd k\}$ and for $j\in\{0\zd i-3\}$. 
However, there is no guarantee that $(d'_{i-2},d'_{i-1})\in\rela$. 
To remedy this we again apply the same construction as follows. Since 
$d_{i-1}$ is a maximal element and $d_{i-1}\sqq d'_{i-1}$, it also holds
$d'_{i-1}\sqq d_{i-1}$. Let $d'_{i-1}=e_{i-1}^1\le\dots\le 
e_{i-1}^\ell=d_{i-1}$ be an s-path connecting $d'_{i-1}$ to $d_{i-1}$.
We set $d''_j=d_j$ for $j\in\{0\zd i-2\}$, and $e^1_j=d'_j$, 
$e^{q+1}_j=e^q_j\cdot e^{q+1}_{i-1}$ and $d''_j=e^\ell_j$ for
$j\in\{i-1\zd k\}$. Then, as before, $(d''_j,d''_{j+1})\in\rela$ for
$j\in\{i-1\zd k-1\}$ and for $j\in\{0\zd i-3\}$. However, since
$d''_{i-2}=d_{i-2}$ and $d''_{i-1}=d_{i-1}$, we also have 
$(d''_{i-2},d''_{i-1})\in\rela$. The result follows.
\end{proof}

\begin{corollary}\label{cor:going-maximal}
Let $\zA\in\cV$ and $\rela$ be a 
tolerance of $\zA$. Suppose that $(a,b)$, $a,b\in\zA$, 
belongs to the transitive closure of $\rela$, that is, there are 
$\vc d{k-1}$ such that $(d_i,d_{i+1})\in\rela$ for 
$i\in\{0,1\zd k-1\}$, where $d_0=a,d_k=b$.
If $a,b$ are maximal in
$\zA$, then there are $d'_0\zd d'_k$ such that $d'_0=a$, each
$d'_j$ is maximal in $\zA$, $(d'_j,d'_{j+1})\in\rela$ for 
$j\in\{0\zd k-1\}$, and $d'_k\in\see\zA b$.
\end{corollary}

\begin{proof}
We show by induction on $i$ that there are $d'_0\zd d'_k$ such that
$d'_0=a$, $(d'_j,d'_{j+1})\in\rela$ for $j\in\{0\zd k-1\}$, 
$d'_k\in\see\zA b$ and $d'_0\zd d'_i$ are maximal. The base case of 
induction, $i=0$, follows from the conditions of the lemma. Suppose 
some $d'_j$s with the required properties exist for $i\in\{0,1\zd k-1\}$. 
Let $d^*_{i+1}$ be a maximal element with $d'_{i+1}\sqq d^*_{i+1}$. 
By Lemma~\ref{lem:going-maximal} there are $d''_{i+1}\zd d''_k$ 
such that $d^*_{i+1}\sqq d''_{i+1}$, $d'_j\sqq d''_j$ for 
$j\in\{i+2\zd k\}$, $(d'_i,d''_{i+1})\in\rela$,
and $(d''_j,d''_{j+1})\in\rela$ for $j\in\{i+1\zd k-1\}$. Elements
$d'_0\zd d'_i,d''_{i+1}\zd d''_k$ satisfy the required conditions.
\end{proof}

The next theorem is the main result of this section. Note that we consider smooth algebras first.

\begin{theorem}\label{the:connectivity}
Let $\zA\in\cK$. 
Any $a,b\in\max(\zA)$ (or $a,b\in\amax(\zA)$, or $a,b\in\umax(\zA)$)
are asm-connected. Moreover if $a,b\in\max(\zA)$ or 
$a,b\in\amax(\zA)$, they are connected by a special path. 
\end{theorem}

\begin{proof}
We start by showing asm-connectivity by a special path (most of the
time we will not mention that we are looking for a special path, except
when it is essential) for maximal elements, so let
$a,b\in\max(\zA)$. We proceed by induction on the size of $\zA$ 
through a sequence of claims. In the base case of induction, when 
$\zA=\{a,b\}$, elements $a,b$ are connected by a thin majority 
or affine edge, as they are both maximal, and the claim is 
straightforward.

\smallskip

{\sc Claim 1.}
$\zA$ can be assumed to be $\Sg{a,b}$.

\smallskip

If $a,b\in\max(\zB)$, $\zB=\Sg{a,b}$, and $\zB\ne\zA$, then we are 
done by the induction hypothesis. Suppose one of them is not 
maximal in $\zB$, and let 
$c,d\in\max(\zB)$ be such that $a\sqq c$ and $b\sqq d$. By the 
induction hypothesis $c$ is asm-connected to $d$. As $a\sqq c$, $a$ is 
asm-connected to $c$. It remains to show that $d$ is asm-connected to $b$. 
This, however, follows straightforwardly from the assumption that 
$b$ is maximal. Indeed, it implies $d\in \see\zA b$, and hence $d\sqq b$ 
in $\zA$.

\smallskip

{\sc Claim 2.}
$\zA$ can be assumed to be simple.

\smallskip

Suppose $\zA$ is not simple and $\al$ is its maximal congruence. Let
$\zB=\zA\fac\al$. By the induction hypothesis $a\fac\al$ is asm-connected to
$b\fac\al$ with a special path, that is, there is a sequence
$a\fac\al=\oa_0,\oa_1\zd \oa_k=b\fac\al$ such that $\oa_i\le\oa_{i+1}$
or $\oa_i\oa_{i+1}$ is a thin affine or special majority edge in $\zB$.
By Corollary~\ref{cor:quotient-path}(1) there is a special path 
$\vc ak$ such that $a_1=a$ and $a_i\in\oa_i$ in $\zA$. 

It remains to show that $a_k$ is asm-connected to $b$. Since $b$ is 
maximal, it suffices to take elements $a',b'$ maximal in $b\fac\al$ 
and such that $a_k\sqq_{\oa_k} a'$ and $b\sqq_{\oa_k} b'$. 
Then $a_k\sqq_{\oa_k} a'$, element $a'$ is asm-connected to $b'$ 
in $\oa_k$ by the induction hypothesis, and $b'$ is 
asm-connected to $b$ in $\zA$, as $b'\in \see\zA b$.

\smallskip

{\sc Claim 3.}
$\Sg{a,b}$ can be assumed to be equal to $\Sg{a',b'}$ for any 
$a'\in\see\zA a$, $b'\in\see\zA b$.

\smallskip

If $\Sg{a',b'}\subset \Sg{a,b}$ for some $a'\in \see\zA a$, 
$b'\in\see\zA b$, then by the induction hypothesis $a''$ is 
asm-connected to $b''$ for some $a''\in \see\zA a$, $b''\in \see\zA b$. 
Therefore $a$ is also asm-connected to $b$.

\smallskip

From now on we assume $\zA$ and $a,b$ to satisfy all the conditions of 
Claims~1--3.

We say that elements $c,d\in\zA$ with $c\in\max(\zA)$ are connected by proper 
subalgebras of $\zA$ if there are subalgebras $\vc\zB\ell$ such
that $\zB_i\ne\zA$ for $i\in[\ell]$, $c\in \zB_1$, $d\in\zB_\ell$, and 
$\zB_i\cap\zB_{i+1}\cap\max(\zA)\ne\eps$ for $i\in[\ell-1]$.

\smallskip

{\sc Claim 4.}
Let $\rel_{ab}$ be the 
subalgebra of $\zA^2$ generated by $(a,b),(b,a)$.
Then 
either $\rel_{ab}$ is the graph of an automorphism $\vf$ of $\zA$ 
such that $\vf(a)=b$, $\vf(b)=a$, or $a,b $ are connected by 
proper subalgebras of $\zA$.

\smallskip

Suppose that $\rel_{ab}$ is not the graph of a mapping, or, 
in other words, there is no automorphism of $\zA$ that maps 
$a$ to $b$ and $b$ to $a$. We consider the link tolerance 
$\rela=\tol_1(\rel_{ab})$. Since $\zA$ is simple and $\rel$ is 
not the graph of a mapping, the transitive closure of $\rela$ 
is the full relation. Suppose first that for every $e\in\zA$ the set 
$\rel^{-1}_{ab}[e]=\{d'\mid (d',e)\in\rel\}$ (which is a 
subalgebra of $\zA$) does not equal $\zA$. There are $\vc ek\in\zA$ 
such that the subalgebras $\vc\zB k$, $\zB_i=\rel^{-1}_{ab}[e_i]$ 
are such that $a\in \zB_1$, $b\in\zB_k$, and 
$\zB_i\cap\zB_{i+1}\ne\eps$ for every $i\in[k-1]$. 
Choose $d_i\in\zB_i\cap\zB_{i+1}$ for $i\in[k-1]$ and note that 
$(d_i,d_{i+1})\in\rela$ and also $(a,d_1),(d_{k-1},b)\in\rela$.
By Corollary~\ref{cor:going-maximal} it is possible to choose 
$\vc{d'}k\in\zA$ such that $(a,d'_1)\in\rela$, 
$(d'_i,d'_{i+1})\in\rela$ for $i\in[k-1]$, all the $d'_i$ are maximal, and
$b\sqq b'=d'_k$. 
The conditions $(a,d'_1)\in\rela$, $(d'_i,d'_{i+1})\in\rela$ mean
that there are $e'_0,\vc{e'}k\in\zA$ such that 
$d'_i,d'_{i+1}\in\rel^{-1}[e'_i]$ for $i\in\{2\zd k\}$ and 
$c,d'_1\in\rel^{-1}_{ab}[e'_1]$. Therefore elements $a,b'$ are connected 
by proper subalgebras of~$\zA$. Since $b$ is maximal, $b'\sqq b$, and so $b'$ is connected to $b$ with proper subalgebras as well, which are the thin semilattice edges in the path from $b'$ to $b$. 

Suppose that there is $e\in\zA$ such that $\zA\tm\{e\}\sse\rel$.
If $e\not\in\max(\zA)$, choose $e'\in\max(\zA)$ with $e\sqq e'$.
By Corollary~\ref{cor:product-path} the path from $e$ to $e'$ can
be extended to paths from $(a,e),(b,e)$ to some $(a',e'),(b',e')$,
respectively. Then, as $a\sqq a'$ and $a\in\max(\zA)$, we also have
$a'\sqq a$. The path from $a'$ to $a$ can again be extended to 
a path from $(a',e')$ to $(a,e'')$ for some $e''$. Also, the path from 
$e'$ to $e''$ can be extended to a path from $(b',e')$ to $(b'',e'')$.
As is easily seen, $b''\in\se b$ and $e''$ is maximal.
Since $\zA=\Sg{a,b''}$ by Claim~3, we have 
$\zA\tm\{e''\}\sse\rel$. Thus, $e$ can be assumed to be maximal. 
We have therefore $(a,b),(a,e),(b,e),(b,a)\in\rel$. If both 
$\Sg{b,e}$ and $\Sg{e,a}$ are proper subalgebras of $\zA$, 
then $a,b$ are connected by proper subalgebras of $\zA$. Otherwise 
suppose $\Sg{b,e}=\zA$. This means $\{a\}\tm\zA\sse\rel_{ab}$, and, 
in particular, $(a,a)\in\rel_{ab}$. Therefore there is a binary term 
operation $f$ such that $f(a,b)=f(b,a)=a$. This means $b\le a$
or $a\le b$, where the latter case is possible if there is also 
another semilattice operation $f'$ on $\{a,b\}$ with $f'(a,b)=f'(b,a)=b$,
and this operation is picked for the relation $\le$. 
If $\Sg{a,e}=\zA$ then by a similar argument we get $a\le b$ or $b\le a$. However, $a\le b$ or $b\le a$ is an impossibility, because in this case $\zA=\{a,b\}$, and only one of $a,b$ is maximal.

\smallskip

{\sc Claim 5.}
Assuming the induction hypothesis, if $c\in\max(\zA)$ and $d\in\zA$ are
connected by proper subalgebras, then $c$ and $d'$ for some $d'$, $d\sqq d'$,
are asm-connected.

\smallskip

There are proper subalgebras $\vc\zB k$ of $\zA$ such that $c\in \zB_1$, 
$d\in\zB_k$, and $\zB_i\cap\zB_{i+1}\cap\max(\zA)\ne\eps$ for every 
$i\in[k-1]$. Let $e_i\in B_i\cap B_{i+1}$ be an element maximal in $\zA$. 
For each $i\in[k-1]$ choose $c_i,c'_i\in\max(\zB_i)$ with 
$e_{i-1}\sqq_{\zB_i} c_i$ and $e_i\sqq_{\zB_i} c'_i$. By the induction 
hypothesis $c_i$ is asm-connected to $c'_i$ (in $\zB_i$). Then clearly 
$e_{i-1}$ is asm-connected to $c_i$, and, as $e_i$ is maximal in $\zA$ and 
$c'_i\in\see\zA{e_i}$, $c'_i$ is asm-connected to $e_i$, as well. The element 
$d'$ is then any with $d'\in\max(\zB_k)$ and $d\sqq d'$.

\smallskip

Let $\rel=\rel_{ab}$. We consider two cases.

\smallskip

{\sc Case 1.}
$\rel$ is not the graph of a mapping.

\smallskip

By Claim~4 
$a,b$ are connected by 
proper subalgebras. By Claim~5 $a$ is asm-connected to some $b'\in\see\zA b$,
which is s-connected to $b$.

\smallskip

{\sc Case 2.}
$\rel$ is the graph of a mapping, or, in other words, there is an
automorphism of $\zA$ that maps $a$ to $b$ and $b$ to $a$.

\smallskip

We again consider several cases.

\smallskip

{\sc Subcase 2a.}
There is no nonmaximal element $c\le a'$ or $c\le b'$ for any
$a'\in\see\zA a$, $b'\in\see\zA b$. In other words, whenever $c\sqq a$
or $c\sqq b$, we have $c\in\see\zA a$ or $c\in\see\zA b$.

\smallskip

Suppose first that there is a maximal $d$ such that $d'\le d$ for some nonmaximal
$d'\in\zA$ (that is, $\zA\ne\max(\zA)$). In this case there is no automorphism of $\zA$ that swaps $a$ and $d$ or $b$ and $d$, because there is no nonmaximal $a'$ [and $b'$] with $a'\le a$ [respectively, $b'\le b$], while $d'\le d$. Therefore we are in the conditions of Case~1, and $a$ is asm-connected to $d$ and $d$ is 
asm-connected to $b$. 

Suppose all elements in $\zA$ are maximal.
By Theorem~\ref{the:connectedness} there are
$a=a_1,a_2\zd a_k=b$ such that for any $i\in[k-1]$ the pair
$a_ia_{i+1}$ or $a_{i+1}a_i$ is a semilattice, affine or majority 
edge (not a thin edge). We need to show that 
$a_i$ is asm-connected
to $a_{i+1}$. Let $\th$ be a congruence of $\zB=\Sg{a_i,a_{i+1}}$
witnessing that $a_ia_{i+1}$ or $a_{i+1}a_i$ is a semilattice, affine, 
or majority edge. Except for the case when $a_{i+1}a_i$ is a 
semilattice edge, by Lemma~\ref{lem:thin-semilattice} there is 
$b\in \zC=a_{i+1}\fac\th$ such that $a_ib$ is a
thin edge. Then take $c,d\in\max(\zC)$ such that $b\sqq_\zC c$
and $a_{i+1}\sqq_\zC d$. By the induction hypothesis $c$ is 
asm-connected to $d$ in $\zC$. Finally, as all elements in $\zA$ are 
maximal, $d$ is asm-connected with $a_{i+1}$ in $\zA$ by a 
semilattice path. If $a_{i+1}a_i$ is a semilattice edge, by 
Lemma~\ref{lem:thin-semilattice} there is $b\in\zD=a_i\fac\th$ such 
that $a_{i+1}\le b$. Since $b\in\see\zA{a_{i+1}}$, there is a 
semilattice path from $b$ to $a_{i+1}$. Then again we choose
$c,d\in\max(\zD)$ with $b\sqq_\zD c$ and $a_i\sqq_\zD d$.
By the induction hypothesis $d$ is asm-connected to $c$.
Since $a_i\sqq_\zD d$ and $c\in\see\zA b$, element $a_i$ is
asm-connected to $a_{i+1}$. Subcase~2a is thus completed.

\smallskip

{\sc Claim 6.}
Let $c\le b$, $c\not\in\max(\zA)$. Then 
$a$ is asm-connected to an element $d$ such that $c\sqq d$. 

\smallskip

If $\Sg{a,c}=\zB\ne\zA$, then there are $a',d\in\zB$ such that 
$a\sqq a'$ and $c\sqq d$ and by the induction hypothesis $a'$ 
is asm-connected to $d$. So, assume that $\Sg{a,c}=\zA$.

Consider the relation $\rel=\rel_{ac}$, that is, the binary relation 
generated by $(a,c)$ and $(c,a)$. Since $a$ is maximal and $c$
is not, there is no automorphism of $\zA$ that swaps $a$ and $c$.
Therefore $\tol_1(\rel)$ is a nontrivial tolerance, in particular,
$(a,b)$ is in its transitive closure. If $\rel^{-1}[e]=\zA$ for no $e\in\zA$,
by Corollary~\ref{cor:going-maximal} $a$ and $b'$ for some 
$b'\in\see\zA b$ are connected by proper subalgebras. By the induction 
hypothesis and Claim~5 $a$ and $b'$ are asm-connected, thus taking 
$d$ to be $b'$ we obtain the claim. So, suppose there is $e\in\zA$ 
such that $\zA\tm\{e\}\sse\rel$. 

By Corollary~\ref{cor:product-path} there are $a'\in\see\zA a$ and 
$e''\in\max(\zA)$ such that $(a',e'')\in\rel$ and $(a,e)\sqq(a',e'')$ in $\rel$. 
Again by Corollary~\ref{cor:product-path} there is $e'\in\max(\zA)$ 
such that $(a,e')\in\rel$ and $(a,e)\sqq(a',e'')\sqq(a,e')$ in $\rel$. Let 
$(a,e)=(a_1,e_1)\le\dots\le(a_k,e_k)=(a,e')$ in $\rel$. 
Consider the sequence $c=c_1\zd c_k$ given by $c_{i+1}=c_i\cdot a_{i+1}$  
for $i\in[k-1]$.
Then, since $(c,e)\in\rel$, we have $(c_1,e_1)\in\rel$, and 
\[
\cl{c_{i+1}}{e_{i+1}}=\cl{c_i}{e_i}\cdot\cl{a_{i+1}}{e_{i+1}}\in\rel,
\qquad \text{for $i\in[k-1]$}
\]
implies that $(a,e'),(c',e')\in\rel$, where $c'=c_k$ and $e'=e_k$. 

We consider several cases. First, suppose that $\zB=\Sg{a,c'}\ne\zA$. 
Then let $a'',c''\in\max(\zB)$ be such that $a\sqq_\zB a''$,
$c'\sqq_\zB c''$. By the induction hypothesis $a''$ is asm-connected 
to $c''$, and, hence, $a$ is asm-connected to $c''$. Since $c\sqq c''$,
take $c''$ for $d$ and the result follows.

Suppose $\Sg{a,c'}=\zA$. Then $(c,e')\in\rel$, as $(a,e'),(c',e')\in\rel$. 
Therefore $(a,c),(a,e'),(c,e'),(c,a)\in\rel$.
If both $\zB_1=\Sg{a,e'}$ and $\zB_2=\Sg{e',c}$ are not equal to $\zA$, 
then $a$ and $c$ are connected with proper subalgebras. By the 
induction hypothesis and Claim~5 $a$ and $d$ are asm-connected 
some $d$ with $c\sqq d$. 
The result holds in this case as well.

Suppose $\Sg{a,e'}=\zA$ or $\Sg{e',c}=\zA$. Then 
$(c,c)\in\rel$ or $(a,a)\in\rel$, which means there is a 
binary operation $f$ such that $f(a,c)=f(c,a)=c$ or 
$f(a,c)=f(c,a)=a$, implying $ac$ is a thin semilattice edge.
Since $c$ is not maximal, we have $c\le a$, and the result follows.

\smallskip

Elements $a,b$ are said to be \emph{v-connected} if there is
$c\in\Sg{a,b}$ such that $c\sqq a$ and $c\sqq b$. 

\smallskip 

{\sc Subcase 2b.}
Elements $a,b$ are v-connected. 

\smallskip

Recall that there is an automorphism of $\zA$ that swaps $a$ and $b$.
The depth of element $c$, denoted $\dep(c)$, is defined to be the 
maximal number of 
s-components on an s-path $c=c_1\le\dots\le c_k$ such that $c_k$ is maximal. 
We proceed  by induction on the size of $\Sg{a,b}$ and minimal $\dep(c)$, 
such that $c\sqq a$, $c\sqq b$. Let $c=a_1\le a_2\le\dots\le a_k=a$ and 
$c=b_1\le b_2\le\ldots\le b_m=b$. 

Suppose first that $k=m=2$, that is, $c\le a$, $c\le b$.
As $c\in\Sg{a,b}$, there is a binary
term operation $f$ such that $f(a,b)=c$. Let $d=f(b,a)$. Since there
is an automorphism swapping $a$ and $b$, $d\le a$ and $d\le b$.
Set $r(x,y,z)=(f(y,x)\cdot f(y,z))\cdot f(x,z)$. We have
\begin{eqnarray*}
r(a,a,b) &=& (ac)c=a,\\
r(a,b,a) &=& (dd)a=a,\\
r(b,a,a) &=& (ca)d=a.
\end{eqnarray*}
Since $a$ and $b$ are swapped by an automorphism, $r$ is a majority operation on
$\{a,b\}$. Therefore, $ab$ is not only a thin majority edge,
but also a majority edge, as is witnessed by the equality relation. 
Since $\zA$ is smooth, $\Sg{a,b}=\{a,b\}$, and this case is in fact
impossible.

Next we consider two cases that cover both the base and inductive steps. 
Let $k>2$ or $m>2$. We may assume $m>2$ 
and $d=b_{m-1}$ to be a nonmaximal element. Indeed, if $b_{m-1}\in\see\zA b$ then we can replace $b$ with $b_{m-1}$.

\smallskip

{\sc Subsubcase I.} $c\not\in\see\zA d$.\\
Suppose the result is proved for all algebras and pairs of elements
v-connected through an element of depth less than $\dep(c)$. 
Note that $\dep(d)<\dep(c)$, because
$d$ is on an s-path from $c$ to a maximal element. 
%% The sets of maximal elements of $\zA$ and its subalgebra may be incomparable. This means that Claim~6 may not apply and that in order to use the inductive hypothesis we need to replace the current element with a maximal one. 
First we observe that there is $e\in\max(\zA)$ such that $a$ is asm-connected to $e$ and $d\sqq_\zA e$. By Claim~6
$a$ is asm-connected to an element $d'$ such that $d\sqq_\zA d'$. 
%% Otherwise let $d'\in\max(\zB)$ be such that $d\sqq_\zB d'$. By the induction hypothesis (on the size of the algebra) $a$ is asm-connected to 
%% $d'$ in $\zB$. 
Then choose $e\in\max(\zA)$ such that $d'\sqq_\zA e$; clearly $a$ is asm-connected to $e$. 
%% . ,
%% . If $\zB=\zA$  Then, as $\dep(d)<\dep(c)$ and $d',b$ are 
%% v-connected through $d$, we use the inductive 
%% hypothesis on the depth of an element through which $a$ and $b$ 
%% are v-connected. In this case $a$ and $b$ are asm-connected.
%% 
%% Suppose $\zB\ne\zA$.  and therefore to $e$ in $\zA$. Also, 
Elements $e$ 
and $b$ are v-connected through $d$. Let $\zC=\Sg{e,b}$. If $\zC=\zA$, the result follows 
by the induction hypothesis, since $\dep(d)<\dep(c)$. Otherwise the induction hypothesis does not apply directly, as $e,b$ may not be maximal in $\zC$ and $\dep(d)$ may change.
We choose $e',b'\in\max(\zC)$ and such that
$e\sqq_\zC e'$, $b\sqq_\zC b'$. By the induction hypothesis 
$e',b'$ are asm-connected in $\zC$, which implies $e,b$ are 
asm-connected in $\zA$, since $e',b'$ belong to $\see\zA e$ and $\see\zA b$, respectively.  

\smallskip

{\sc Subsubcase II.} $c\in\see\zA d$. This covers in particular the base case.\\
Let $\zB=\Sg{a,d}$. 
We show that there is $e\in\zB$ such that $a$ is asm-connected to $e$ and
either $e\in\see\zA d$ (this includes $e=d$) or $d\le e$. Observe that this
implies that $a$ is asm-connected to $b$. Indeed, if $e\in\see\zA d$ then $a$
is asm-connected to $e$, which is asm-connected to $b$. Note that 
this includes the case $e=a$. Otherwise if 
$e\in\max(\zA)$, as $d\le e$ and $d\le b$, we argue as in the beginning of 
Subcase~2b. Suppose $e\not\in\max(\zA)$. Take any $e'\in\max(\zA)$ 
with $e\sqq_\zA e'$, $a$ is asm-connected to $e'$. For $b$ and $e'$ we consider 
connections $d\sqq b$ and $d\le e\le\dots\le e'$. Since 
$e\not\in\see\zA d$, we have $\dep(e)<\dep(d)$, and we are in the 
conditions of Subsubcase~I (with $d$ playing the role of $c$, $e'$ and $b$ playing the role of $b$ and $a$, respectively and some element on the s-path from $e$ to $e'$ playing the role of $d$). 

Suppose first that for some $a'\in\see\zA a$, $\Sg{a',d}\ne\zA$. We may assume that $a'=a$. Then by the induction hypothesis, $a$ is asm-connected to all maximal elements of $\zB$. 
Let $\zC$ be a minimal subalgebra of $\zB$ of the form $\Sg{d,e''}$, where 
$a$ is asm-connected to $e''$ and $d\sqq_\zB e''$. Take $e\in\max(\zC)$ such that
$d\sqq_\zC e$. Then by the inductive hypothesis 
$a$ is asm-connected to $e$. If $e\in\see\zA d$, we are done. Otherwise 
$\Sg{d,e}=\zC$ by the choice of $e''$. We show that $d\le e$. Consider the relation 
$\relo$ generated by $(e,d),(d,e)$. It suffices to show that $(e,e)\in\relo$. 
Let $d=d_1\le d_2\le\dots\le d_x=e$ be an s-path in $\zC$, and 
$(d,e)=(d_1,e_1)\le(d_2,e_2)\le\dots\le(d_x,e_x)$ its extension in $\relo$. Then
$d_x=e$, $e_x\in\see\zC e$. By the choice of $\zC$ and $e$, $e\in\Sg{d,e_x}=\zC$. 
Therefore the pairs $(e,d),(e,e_x)\in\relo$ generate $(e,e)$.

Suppose now that $\Sg{a',d}=\zA$ for all $a'\in\see\zA a$. Then $d\sqq_\zB a$, 
because $d\sqq_\zA c$, and by the above argument $d\le a$.

{\sc Subcase 2c.}
Elements $a,b$ are not v-connected

\smallskip

Note first that we may assume that, for any $b'\in\see\zA b$, there is an
automorphism that sends $b'$ to $a$ and $a$ to $b'$, as otherwise we
are in the conditions of Case~1. Recall that we also assume
$\Sg{a,b'}=\zA$. Because of this and the automorphism swapping $a$
and $b$, without loss of generality we may assume that there is
nonmaximal $c\le b$. %% Consider $\Sg{a,c}$.
%% 
%% If $\Sg{a,c}=\zA$, 
By Claim~6 
$a$ is asm-connected to some $d$ such that $d$ is v-connected with $b$. 
The result follows from Subcase~2b.
%% 
%% If $\zB=\Sg{a,c}\ne\zA$, take $d\in\max(\zB)$ and such that $c\sqq d$. By the induction hypothesis $a$ is asm-connected to $d$. Now let $d\sqq d'$ such that $d'\in\max(\zA)$. Since $d'$ is v-connected to $b$, the result again follows from Subcase~2b.

\medskip

In the remaining statements of the theorem, when $a,b\in\amax(\zA)$ or
$a,b\in\umax(\zA)$, we let $a',b'\in\zA$ be maximal elements of 
$\zA$ such that $a\sqq a'$ and $b\sqq b'$. Then by what is proved 
above $a'$ is asm-connected to $b'$, and so $a$ is asm-connected to $b'$. 
Finally, as $b'\in\as(b)$ [respectively, $b'\in\umax(\zA)$],
we have $b'\sqq_{as}b$ [respectively, $b'\sqq_{asm}b$], and
$b'$ is connected to $b$. Note that in the case of u-maximal elements 
the path showing that $b'\sqq_{asm}b$ cannot be assumed to be special.
\end{proof}

Theorem~\ref{the:connectivity} can be extended to algebras from the 
variety $\cV$ generated by $\cK$.

\begin{corollary}\label{cor:var-connectivity}
Let $\zA\in\cV$. 
Any $a,b\in\max(\zA)$ (or $a,b\in\amax(\zA)$, or $a,b\in\umax(\zA)$)
are asm-connected. 
\end{corollary}

\begin{proof}
We need to prove that the property in the corollary is true for subalgebras of 
direct products of algebras from $\cK$ and is preserved by taking quotient 
algebras.

Let $\zA\in\cV$ be such that any $a,b\in\max(\zA)$ (or $a,b\in\amax(\zA)$, or 
$a,b\in\umax(\zA)$) are asm-connected, and $\th\in\Con(\zA)$. Take any 
maximal (as-maximal, u-maximal) $\oa,\ob\in\zA\fac\th$. 
By Corollary~\ref{cor:quotient-max}(1)
there are $a\in\oa,b\in\ob$ that are maximal (as-maximal, u-maximal) 
in $\zA$. Then there
is an asm-path $a=a_1\zd a_k=b$ connecting $a$ and $b$ in $\zA$. By
Corollary~\ref{cor:quotient-path}(1) $\oa=a_1\fac\th\zd a_n\fac\th=\ob$ is
an asm-path in $\zA\fac\th$.

Now let $\rel$ be a subalgebra of the direct product of $\vc\zA n\in\cK$. Since
$\cK$ is closed under taking subalgebras, we may assume that $\rel$ is subdirect.
We proceed by induction on $n$. The base case $n=1$ is given by 
Theorem~\ref{the:connectivity}, suppose the claim is true for $n-1$. Take
maximal (as-maximal, u-maximal) $\ba,\bb\in\rel$. By 
Corollary~\ref{cor:product-maximal}(2) 
$\ba'=\pr_{[n-1]}\ba,\bb'=\pr_{[n-1]}\bb$ are maximal (as-maximal, 
u-maximal) in 
$\pr_{[n-1]}\rel$. By the induction hypothesis $\ba',\bb'$ are connected with an 
asm-path $\ba'=\ba'_1\zd\ba'_k=\bb'$. By Corollary~\ref{cor:product-path}(1) 
this path can be expanded to an asm-path 
$\ba=(\ba'_1,\ba[n])=(\ba'_1,a_1)\zd(\ba'_k,a_k)=(\bb',a_k)$. Consider 
$B=\{b\mid (\bb',b)\in\rel\}$. This set contains $a_k$ and $\bb[n]$. 
Choose maximal elements $c,d\in B$ with $a_k\sqq_Bc$ and $\bb[n]\sqq_Bd$.
By Theorem~\ref{the:connectivity} $a_k$ is asm-connected to $d$ implying
that $\ba$ is asm-connected to $(\bb',d)$ in $\rel$. Then also $\bb\sqq(\bb',d)$,
and, as $\bb$ is maximal (as-maximal, u-maximal), $(\bb',d)\sqq\bb$ (respectively, 
$(\bb',d)\sqq^{as}\bb$, $(\bb',d)\sqq^{asm}\bb$.
\end{proof}

As by Theorem~\ref{the:connectivity} any u-maximal elements 
are connected by an asm-path, we have the following 

\begin{corollary}\label{cor:unique-umax}
Any algebra $\zA\in\cV$ has a unique u-maximal component.
\end{corollary}

%%%%%%%%%%%%%%%%%%%%%%%%%%%%%%%%%%
%%%%%%%%%%%%%%%%%%%%%%%%%%%%%%%%%%
\section{Rectangularity}\label{sec:rectangularity}

In this section we prove a result that, on one hand, is a 
generalization of the results by the author 
\cite{Bulatov06:semilattice} (Lemma~3.5), \cite{Bulatov11:conservative} 
(Lemma~4.2, Proposition~5.4), and, on the other
hand, is analogous to the Rectangularity Lemma from \cite{Barto12:absorbing}.

We will need two auxiliary lemmas.

\begin{lemma}\label{lem:as-rectangularity} 
Let $\rel$ be a subalgebra of $\zA_1\tm\zA_2$, $\zA_1,\zA_2\in\cV$.\\
(1) Let $a,b\in\zA_1$, $c,d\in\zA_2$ 
be such that $(a,c),(a,d),(b,c)\in\rel$ and $ab$, $cd$ are thin edges that are not 
both thin majority edges. Then $(b,d)\in\rel$.\\ 
(2) Let $B=\rel[a]$. For any $b\in\zA_1$ such that $ab$ is thin edge, 
and any $c\in\rel[b]\cap B$, $\Filt^{as}_B(c)\sse\rel[b]$.
\end{lemma}

\begin{proof}
(1) Without loss of generality assume that $cd$ is not majority. Suppose that $ab$ 
is semilattice or majority, or $ab$ is affine and $cd$ is semilattice. Then by 
Lemma~\ref{lem:thin-combination-variety}(2) there is a term operation $p$ 
such that $p(a,b)=b$ and $p(d,c)=d$ (if both $ab$ and $cd$
are of the semilattice type then $p$ can be chosen to be $\cdot$). Then 
\[
\cl bd=p\left(\cl ad,\cl bc\right)\in\rel.
\]

If both $ab$ and $cd$ are affine, then by 
Lemma~\ref{lem:thin-combination-variety}(1)
there is a term operation $h'$ such that $h'(a,a,b)=b$ and 
$h'(d,c,c)=d$. Then
\[
\cl bd=h'\left(\cl ad,\cl ac,\cl bc\right)\in\rel.
\] 

(2) Let $D=\Filt^{as}_B(c)\cap\rel[b]$. The set $D$ is nonempty, as 
$c\in D$. If $D\ne \Filt^{as}_B(c)$, there are $b_1\in D$ 
and $b_2\in\Filt^{as}_B(c)-D$ such that $b_1b_2$ is a thin semilattice
or affine edge. By item (1) $(b,b_2)\in\rel$, a contradiction with the choice of $b_2$.
\end{proof}

\begin{lemma}\label{lem:buket}
Let $\rel$ be a subdirect product of algebras $\zA_1,\zA_2\in\cV$, let
$B_1,B_2$ be as-components (maximal components)
of $\zA_1,\zA_2$, respectively, and $a\in\zA_1$ such that 
$\rel\cap(B_1\tm B_2)\ne\eps$
and $\{a\}\tm B_2\sse\rel$. Then $B_1\tm B_2\sse\rel$.
\end{lemma}

\begin{proof}
We prove the lemma for as-components; for maximal components 
the proof is nearly identical.

Let $(b,c)\in\rel\cap(B_1\tm B_2)\ne\eps$. For every 
$b'\in\zA'_1=\Sg{a,b}$ we have
$(b',c)\in\rel$. By Lemma~\ref{lem:as-rectangularity} this means that 
$\{b'\}\tm B_2\sse\rel$ for all $b'\in\Filt^{asm}_{\zA'_1}(a)$. Indeed,
let $C$ be the set of all elements $b'$ from $\Filt^{asm}_{\zA'_1}(a)$
such that $\{b'\}\tm B_2\sse\rel$. The set $C$ is nonempty, as $a\in C$.
If $C\ne\Filt^{asm}_{\zA'_1}(a)$, there are $b'\in C$ and 
$b''\in\Filt^{asm}_{\zA'_1}(a)-C$ such that $b'b''$ is a thin edge. Then 
by Lemma~\ref{lem:as-rectangularity} $\Filt^{as}_{\rel[b']}(c)\sse\rel[b'']$.
Since $B_2\sse\Filt^{as}_{\rel[b']}(c)$, we have a contradiction.

By Theorem~\ref{the:connectivity} there is an asm-path from $a$ to every 
maximal element from $\zA'_1$. Therefore, the set $\Filt^{asm}_{\zA'_1}(a)$ 
contains all the maximal elements of $\zA'_1$ 
including some $b'$ such that $b\sqq^{as}_{\zA'_1}b'$. Thus, $\{b'\}\tm B_2\sse\rel$ 
for some $b'\in B_1$. By Corollary~\ref{cor:product-path}(1) for every 
$b''\in B_1$ we have $\rel[b'']\cap B_2\ne\eps$. Then 
Lemma~\ref{lem:as-rectangularity} implies that $B_2\sse\rel[b'']$ for 
every $b''\in B_1$.
\end{proof}

\begin{prop}\label{pro:max-gen}
Let $\rel\le\zA_1\tm \zA_2$, $\zA_1,\zA_2\in\cV$, be a linked subdirect 
product and let 
$B_1,B_2$ be as-components (maximal components) of $\zA_1,\zA_2$, respectively, such that $\rel\cap(B_1\tm B_2)\ne\eps$. Then $B_1\tm B_2\sse\rel$.
\end{prop}

\begin{proof}
We prove by induction on the size of $\zA_1,\zA_2$ that for any 
as-components $C_1,C_2$ of $\zA_1,\zA_2$, respectively, such that 
$\rel\cap(C_1\tm C_2)\ne\eps$, there are 
$a_1\in\zA_1$, $a_2\in\zA_2$ 
such that $\{a_1\}\tm C_2\sse\rel$ and $C_1\tm\{a_2\}\sse\rel$.
The result then follows by Lemma~\ref{lem:buket}. The base
case of induction when $|\zA_1|=1$ or $|\zA_2|=1$ is obvious.

Take $b\in C_1$ and construct two sequences
of subalgebras $\vc\zB k$ of $\zA_1$ and $\vc\zC k$ of $\zA_2$, where
$\zB_1=\{b\}$, $\zC_i=\rel[\zB_i]$, and $\zB_i=\rel^{-1}[\zC_{i-1}]$, 
such that $k$ is the minimal number with $\zB_k=\zA_1$ or 
$\zC_k=\zA_2$. Such a number exists, because $\rel$ is linked. Observe 
that for each $i\le k$ the relation $\rel_i=\rel\cap(\zB_i\tm\zC_i)$ is linked. 
Therefore, there is a proper subalgebra $\zA'_1$ of $\zA_1$ or $\zA'_2$ 
of $\zA_2$ such that $\rel'=\rel\cap(\zA'_1\tm\zA_2)$ or 
$\rel'=\rel\cap(\zA_1\tm\zA'_2)$, respectively, is linked and subdirect.
Without loss of generality suppose there is $\zA'_1$ with the required 
properties. By the induction hypothesis for any as-component $C_2$ 
of $\zA_2$ there is $a_1\in\zA'_1\sse\zA_1$ with
$\{a_1\}\tm C_2\sse\rel'\sse\rel$. By Lemma~\ref{lem:buket} $C_1\tm C_2\sse\rel$,
and therefore for any $a_2\in C_2$ we have $C_1\tm\{a_2\}\sse\rel$.

The argument above also works in the case when $B_1,B_2$ are maximal components.
\end{proof}

\begin{corollary}\label{cor:linkage-rectangularity}
Let $\rel$ be a subdirect product of $\zA_1$ and $\zA_2$ from $\cV$, 
$\lnk_1,\lnk_2$ the link congruences, and let $B_1,B_2$ be 
as-components (maximal components) of an $\lnk_1$-block and an $\lnk_2$-block, respectively, 
such that $\rel\cap(B_1\tm B_2)\ne\eps$. Then $B_1\tm B_2\sse\rel$.
\end{corollary}

\begin{proof}
Let $C_1,C_2$ be the $\lnk_1$- and $\lnk_2$-blocks containing $B_1$ 
and $B_2$, respectively, and $\relo=(C_1\tm C_2)\cap\rel$. By definition 
$\relo$ is a subdirect product of $C_1\tm C_2$, as $\rel$ is subdirect, 
and $\relo$ is linked. The result follows by Proposition~\ref{pro:max-gen}.
\end{proof}

\begin{prop}\label{pro:umax-rectangular}
Let $\rel$ be a subdirect product of $\zA_1,\zA_2\in\cV$, 
$\lnk_1,\lnk_2$ the link
congruences, and let $B_1$ be an as-component of an $\lnk_1$-block and
$B'_2=\rel[B_1]$; let $B_2=\umax(B'_2)$. Then $B_1\tm B_2\sse\rel$.
\end{prop}

\begin{proof}
Let $B'_2$ be a subset of a $\lnk_2$-block $C$. By 
Corollary~\ref{cor:product-maximal}(1) for any $a_0\in B'_2$ the set $B'_2$ contains an as-maximal 
element $a$ of $C$ such that $a_0\sqq^{as}a$. By Corollary~\ref{cor:linkage-rectangularity} 
$B_1\tm\{a\}\sse\rel$. It then suffices to show that 
$B_1\tm\Filt^{asm}_{B'_2}(a)\sse\rel$.

Suppose for $D\sse\Filt^{asm}_{B'_2}(a)$ it holds $B_1\tm D\sse\rel$.
If $D\ne \Filt^{asm}_{B'_2}(a)$, there are $b_1\in D$ and
$b_2\in\Filt^{asm}_{B'_2}(a)-D$ such that $b_1b_2$ is a thin edge. 
By Lemma~\ref{lem:as-rectangularity} $B_1\tm\{b_2\}\sse\rel$; the 
result follows.
\end{proof}

%%%%%%%%%%%%%%%%%%%%%%%%%%%%%%%%%
%%%%%%%%%%%%%%%%%%%%%%%%%%%%%%%%%
\section{Quasi-2-decomposability}\label{sec:quasi-2-decomposability}

%% In this section (except for Corollary~\ref{cor:quasi-2-decomp-variety} and 
%% Theorem~\ref{the:pseudo-majority}) algebras are assumed to be smooth 
%% from $\cK$.

An ($n$-ary) relation over a set $A$ is called
\emph{2-decomposable} if, for any tuple $\ba\in A^n$, $\ba\in\rel$ if
and only if, for any $i,j\in[n]$, $\pr_{ij}\ba\in\pr_{ij}\rel$ 
\cite{Baker75:chinese-remainder,Jeavons98:consist}. 2-decomposability 
is closely related to
the existence of majority polymorphisms of the relation. 
Relations over general smooth algebras do not have a majority 
polymorphism, but they
still have a property close to 2-decomposability. We say that a
relation $\rel$, a subdirect product of $\vc{\zA}n\in\cV$, is
\emph{quasi-2-decomposable}, if for any elements $\vc an$,
such that $(a_i,a_j)\in\pr_{ij}\rel$ for any
$i,j$, there is a tuple $\bb\in\rel$ with $(a_i,a_j)\sqq^{as}(\bb[i],\bb[j])$
for all $i,j\in[n]$. In particular, if $(a_i,a_j)\in\amax(\pr_{ij}\rel)$ for any
$i,j$, then $\bb\in\rel$ can be chosen such that $(\bb[i],\bb[j])\in\as{(a_i,a_j)}$,
$i,j\in[n]$.

\begin{theorem}\label{the:quasi-2-decomp}
Let $\vc\zA n\in\cK$. Then any subdirect product $\rel$ of
$\tms\zA n$ is quasi-2-decomposable.

Moreover, if $J$ is a collection of subsets of $[n]$ containing all the 
2-element subsets, $X\in J$, tuple $\ba$ is such that 
$\pr_Y\ba\in\pr_Y\rel$ for every $Y\in J$, and $\pr_X\ba
\in\amax(\pr_X\rel)$, there is a tuple $\bb\in\rel$ with
$\pr_Y\ba\sqq^{as}\pr_Y\bb$ for $Y\in J$, and
$\pr_X\bb=\pr_X\ba$. 
\end{theorem}

%%%%%%%%%%%%%%%%%%%%%%%%%%%%%%%%%
\subsection{The ternary case}

We start with ternary relations assuming every $(a_i,a_j)$ is as-maximal.

\begin{lemma}\label{lem:3-quasi}
Let $\rel$ be a subalgebra of $\zA_1\tm\zA_2\tm\zA_3$, and let 
$(a_1,a_2,a_3)$ be such that $(a_i,a_j)\in\amax(\pr_{ij}\rel)$ for 
$i,j\in\{1,2,3\}$, $i\ne j$. Then there is $(a'_1,a'_2,a'_3)\in\rel$ such that 
$(a'_i,a'_j)$ is in the as-component of $\pr_{ij}\rel$ containing $(a_i,a_j)$ for 
$i,j\in\{1,2,3\}$, $i\ne j$.
\end{lemma}

\begin{proof}
By replacing $\zA_i$ with $\pr_I\rel$ for $i=1,2,3$, the algebra $\rel$ can be 
assumed to be a subdirect product of $\zA_1\tm\zA_2\tm\zA_3$.
We proceed by induction on the size of $\zA_1,\zA_2,\zA_3$. The base case 
of induction is comprised of the following situations: (1) for some $i\in[3]$, 
$|\zA_i|=2$ and $\zA_i$ is a semilattice edge; (2) for some $i\in[3]$, 
$\zA_i$ is a module (not necessarily 2-element); and (3) $|\zA_i|=2$ and $\zA_i$ is a 
majority edge for all $i\in[3]$. By the assumption 
some tuples $\ba_1=(b_1,a_2,a_3)$, $\ba_2=(a_1,b_2,a_3)$, 
$\ba_3=(a_1,a_2,b_3)$ belong to $\rel$. If one of $\zA_1,\zA_2,\zA_3$ 
is a semilattice edge, say, $b_1\le a_1$, then from the as-maximality of 
$a_1,a_2,a_3$, we obtain $(a_1,a'_2,a_3)=\ba_1\cdot\ba_2\in\rel$,
$a'_2=a_2\cdot b_2$. As is easily seen, this tuple satisfies the requirements of 
the lemma. If one 
of $\zA_1,\zA_2,\zA_3$ is a module, say, $\zA_1$ is, then $\ba_1$ satisfies 
the requirements of the lemma. Finally, if all $\zA_1,\zA_2,\zA_3$ are 
majority edges, then $(a_1,a_2,a_3)=g(\ba_1,\ba_2,\ba_3)$, where $g$ is 
the operation from Theorem~\ref{the:uniform}. 

Suppose that the lemma is proved for any subdirect product of 
$\zA'_1\tm\zA'_2\tm\zA'_3$, where $\zA'_i$ is a subalgebra or a factor 
of $\zA_i$, $i\in[3]$, and at least one of them is a proper subalgebra 
or a factor. Let $\ba_1,\ba_2,\ba_3\in\rel$ be as before. Also let $\cD$ 
denote the set of $(c_1,c_2,c_3)\in\zA_1\tm\zA_2\tm\zA_3$ such that
$(c_i,c_j)$, $i,j\in[3]$, belongs to the as-component of $\pr_{ij}\rel$ 
containing $(a_i,a_j)$, for $i\ne j$. The set $\cD$ is nonempty, as 
$\ba=(a_1,a_2,a_3)\in\cD$.

\smallskip

{\sc Claim 1.}
Every $\zA_i$ can be assumed to be $\Sg{a_i,b_i}$, and $b_i$ can be 
chosen to be an as-maximal element.

\smallskip

Suppose $\zA_1\ne\zB=\Sg{a_1,b_1}$. Let $(a'_1,a'_2)$ be an 
as-maximal element in $\relo=(\zB\tm\zA_2)\cap\pr_{12}\rel $ such 
that $(a_1,a_2)\sqq^{as} (a'_1,a'_2)$ in $\relo$. Let also 
$(a_1,a_2)=(c^1_1,c^1_2)\zd(c^k_1,c^k_2)=(a'_1,a'_2)$ be an
as-path from $(a_1,a_2)$ to $(a'_1,a'_2)$ in $\relo$. By 
Corollary~\ref{cor:product-path}(1) it can be extended to an as-path 
$\bc_1\zd\bc_k\in\rel$ with $\bc_1=(a_1,a_2,b_3)$ in 
$\rel'=(\zB\tm\zA_2\tm\zA_3)\cap\rel$. Using 
Lemma~\ref{lem:product-edge}(1) we 
define a sequence $\vc{\bc'}k$ in $\cD$ as follows: $\bc'_1=\ba$, and 
$\bc'_{i+1}$ is such that 
$(\bc'_{i+1}[1],\bc'_{i+1}[2])=(c^{i+1}_1,c^{i+1}_2)$ and 
$\bc'_i\bc'_{i+1}$ is a thin semilattice or affine edge in 
$\Sgg{\zA_1\tm\zA_2\tm\zA_3}{\bc'_i,\bc_{i+1}}$. Note that 
$(\bc'_i[u],\bc'_i[v])(\bc'_{i+1}[u],\bc'_{i+1}[v])$ is a thin edge of 
the same type as $\bc'_i\bc'_{i+1}$ for $u,v\in[3]$. In particular, 
$(a_u,a_v)\sqq^{as}(\bc'_{i+1}[u],\bc'_{i+1}[v])$ in $\pr_{uv}\rel'$. 
Replace $\ba$ with $\bc'_k$. Repeating the process for the other binary 
projections if necessary we obtain $(a'_1,a'_2,a'_3)\in\cD$ such 
that $(a'_i,a'_j)$ is as-maximal in $\pr_{ij}\rel'$. By the induction 
hypothesis there is 
$(a''_1,a''_2,a''_3)\in\rel\cap(\zB\tm\zA_2\tm\zA_3)$ 
such that $(a''_i,a''_j)$ is in the as-maximal component containing 
$(a'_i,a'_j)$. Clearly, $(a''_1,a''_2,a''_3)$ is as required.

If, say, $b_1$ is not an as-maximal element, then choose an as-path 
$b_1=c_1\le\dots c_k$ and its extension $\bc_1\zd\bc_k\in\rel$, $\bc_1=\ba_1$, 
such that $c_k$ is a maximal element and $\bc_k=(c_k,a'_2,a'_3)$. 
Then we choose an as-path in $\pr_{23}\rel$ 
from $(a'_2,a'_3)$ to $(a_2,a_3)$. Extending this path as before 
we get $(d,a_2,a_3)\in\rel$ such that $d\in\amax(\zA_1)$.

\smallskip

{\sc Claim 2.}
For every $i,j\in[3]$, it holds that $\as(a_i)\tm\zA_j\sse\pr_{ij}\rel$ and 
$\as(a_i)\tm\as(a_j),\as(a_i)\tm\as(b_j)$ are as-components
of $\pr_{ij}\rel$.

\smallskip

Since $(a_i,a_j),(a_i,b_j)\in\pr_{ij}\rel$ and $\zA_j=\Sg{a_j,b_j}$, 
we have $\{a_i\}\tm\zA_j\sse\pr_{ij}\rel$. By Lemma~\ref{lem:buket} 
$\as(a_i)\tm\as(a_j),\as(a_i)\tm\as(b_j)\sse\pr_{ij}\rel$. Therefore, for any 
$c\in\as(a_i)$ it holds that $(c,a_j),(c,b_j)\in\pr_{ij}\rel$, implying $\as(a_i)\tm\zA_j\sse\pr_{ij}\rel$. The second statement is obvious. 

\smallskip

{\sc Claim 3.}
Every $\zA_i$ can be assumed to be simple.

\smallskip

Suppose $\th$ is a nontrivial congruence of $\zA_1$ and 
$\rel\fac\th=\{(c_1\fac\th,c_2,c_3)\mid\lb (c_1,c_2,c_3)\in\rel\}$. By the 
induction hypothesis there is $(a''_1,a'_2,a'_3)\in\rel\fac\th$ satisfying 
the conditions of the lemma, that is, there is $(b_1,a'_2,a'_3)\in\rel$ 
such that $b_1\fac\th=a''_1$, and $(a_2,a_3)\sqq^{as}(a'_2,a'_3)$, 
$(a_1\fac\th,a_i)\sqq^{as}(a''_1,a'_i)$ for $i\in\{2,3\}$, where the latter
as-paths are in $\pr_{1i}\rel\fac\th$. Let $a'_1\in b_1\fac\th$ be any 
element such that $a_1\sqq^{as} a'_1$ and $a'_1$ is maximal 
in $b_1\fac\th$. Such an element exists, because 
$a_1\fac\th\sqq^{as}b_1\fac\th$. Then for $(a'_1,a'_2,a'_3)$ we have 
$(a'_i,a'_j)\in\as(a_i)\tm\as(a_j)\sse\pr_{ij}\rel$, for any 
$i,j\in\{1,2,3\}$, where the last inclusion is by Claim~2.  
Therefore $(a_i,a_j)\sqq^{as}(a'_i,a'_j)$. Since 
$\Sg{a'_1,b_1}\ne\zA_1$, the claim follows by the induction hypothesis.

\smallskip

We now prove the induction step. Suppose now that $|\zA_i|>2$ 
for some $i$ and $\zA_i$ is not a 
module for any $i$. For an $n$-ary relation $\relo\le\tms\zA n$,
$j\in[n]$, and $c_j\in\zA_j$, let $\relo[c_j]$ denote the set 
$\{(c_1\zd c_{j-1},c_{j+1}\zd c_n)\in\pr_{\{1\zd j-1,j+1\zd n\}}\relo\mid 
(\vc cn)\in\relo\}$. We still use the tuples $\ba_1,\ba_2,\ba_3\in\rel$. 
There are two cases to consider. 

\smallskip

{\sc Case 1.} For some $i\in[3]$ the set $\rel[b_i]$ contains 
$\as(a_j)\tm\as(a_\ell)$, where $\{j,\ell\}=[3]-\{i\}$.

\smallskip

Assume $i=1$. Since $a_1$ is as-maximal, by Theorem~\ref{the:connectivity} 
there is a special asm-path $P$ from $b_1$ to $a_1$. We prove that 
for any element $c$ on this path $\{c\}\tm\as(a_2)\tm\as(a_3)\sse\rel$. 
This is true for $c=b_1$ by the assumption made. Assume the 
contrary, and let $c$ be the first element in $P$ for which this 
property is not true. Let also $d$ be the element preceding $c$ in 
$P$; we may assume $d=b_1$. If $b_1c$ is semilattice or affine, 
then by Lemma~\ref{lem:product-edge}(1)
applied to $(b_1,a_2,a_3)$ and $c$ there is $\bc=(c,a'_2,a'_3)\in\rel$ 
such that $(a'_2,a'_3)\in\as(a_2,a_3)$. Therefore by 
Lemma~\ref{lem:as-rectangularity} $\{c\}\tm\as(a_2)\tm\as(a_3)\sse\rel$.

Let $b_1c$ be a special thin majority edge, $\zB=\Sg{b_1,c}$, and $\th$ 
a congruence witnessing that $b_1c$ a majority edge; in particular, 
$\zB=b_1\fac\th\cup c\fac\th$, as $\zA_1$ is smooth. Suppose first that 
$\zB=\zA_1$. Then $\th$ is the equality relation, as $\zA_1$ is simple, 
and so $|\zA_1|=2$. In this case $c=a_1$ and $\as(a_1)=\{a_1\}$. 
We prove that $\as(a_1)\tm\as(a_3)\sse\rel[b_2]$. By Claim~2 
$\as(b_2)\tm\as(a_3)\sse\pr_{23}\rel$ and $(a_1,b_2,a_3)\in\rel$.
Extending an as-path in $\pr_{23}\rel$ from $(b_2,a_3)$ to an arbitrary
$(d_2,d_3)\in\as(b_2)\tm\as(a_3)$ Corollary~\ref{cor:product-path} implies 
that $(a_1,d_2,d_3)\in\rel$. A similar
argument shows that $\as(a_1)\tm\as(a_2)\sse\rel[b_3]$. Therefore,
either $|\zA_1|=|\zA_2|=|\zA_3|=2$, which is the base case, or we may
assume that $\zB\ne\zA_1$.

Consider $\rel'=\rel\cap(\zB\tm\zA_2\tm\zA_3)$. 
%%, by Claim~2 it is a subdirect product of $\zB,\zA_2,\zA_3$.
Take any 
$e\in\amax(\zB)\cap c\fac\th$. Such an element exists, because by 
Corollary~\ref{cor:quotient-path}(2) any as-path that starts in $c\fac\th$
remains inside $c\fac\th$. For the tuple $(e,a_2,a_3)$ we have 
the following. 
By Claim~2 $(e,a_i)\in\pr_{1i}\rel'$ for 
$i\in\{2,3\}$. Also, $(a_2,a_3)\in\pr_{23}\rel'$, as 
$(b_1,a_2,a_3)\in\rel$ by the assumption made. By the induction 
hypothesis there is $(e',a'_2,a'_3)\in\rel'$ with $e'\in\as_\zB(e)$ 
(and so $e'\in c\fac\th$) and $a'_i\in\as_{\zA_i}(a_i)$, $i\in\{2,3\}$.
Let $e''=g(b_1,e',e')$, where $g$ is the operation satisfying the majority 
condition from Theorem~\ref{the:uniform}. 
Then $g(b_1,e'',e'')=e''$. Since $b_1c$ is a minimal pair with respect 
to $\th$, it holds that $c\in\Sg{b_1,e''}$. By 
Lemma~\ref{lem:thin-semilattice}
$b_1e''$ is also a thin majority edge. Moreover, 
\[
\cll{e''}{a'_2}{a'_3}=g\left(\cll{b_1}{a'_2}{a'_3},
\cll{e'}{a'_2}{a'_3},\cll{e'}{a'_2}{a'_3}\right)\in\rel'.
\]
By Lemma~\ref{lem:as-rectangularity} 
\[
\as(a_2)\tm\as(a_3)\sse\Filt^{as}_{\pr_{23}\rel'}((a'_2,a'_3))
\sse\rel[e''].
\]

Since $\Sg{b_1,e''}=\Sg{b_1,c}$, $c=r(b_1,e'')$ for some term 
operation $r$. It remains to notice that
\[
\cll c{a''_2}{a''_3}=r\left(\cll{b_1}{a''_2}{a''_3},
\cll{e''}{a''_2}{a''_3}\right)\in\rel'
\]
for any $a''_2\in\as(a_2)$, $a'' _3\in\as(a_3)$, a contradiction 
with the choice of $c$.

\smallskip

{\sc Case 2.} For all $i\in[3]$, $\as(a_j)\tm\as(a_\ell)\not\sse\rel[b_i]$, 
where $\{j,\ell\}=[3]-\{i\}$.

\smallskip

Let $\lnk_{j\ell}$ be the link congruence of $\pr_{j\ell}\rel$ when 
$\rel$ is viewed as a subdirect product of $\zA_i$ and 
$\pr_{j\ell}\rel$; and let $\lnk_i$ be the link congruence of $\zA_i$. 
Since $b_i$ is as-maximal, if $\lnk_i$ is the total congruence, then 
by Proposition~\ref{pro:max-gen} and Claim~2 $\as(a_j)\tm\as(a_\ell)\sse\rel[b_i]$, a 
contradiction with the assumption made. Therefore $\lnk_i$ is the equality 
relation for all $i\in[3]$. Consider the $\lnk_{j\ell}$-block
$\relo=\rel[a_i]$. By Claim~2 $\relo$ is a subdirect product of 
$\zA_j\tm\zA_\ell$. If $\relo$ is linked, $\as(b_j)\tm\as(a_\ell)\sse\relo$.
In this case, if $|\as(a_i)|=1$ then $\as(a_i)\tm\as(a_\ell)\sse\rel[b_j]$,
a contradiction with the assumptions of Case~2. If $|\as(a_i)|>1$, for
any $c_i\in\zA_i$ such that $a_ic_i$ is a thin semilattice or affine edge
by Lemma~\ref{lem:product-edge}(1)
$(c_j,c_\ell)\in\rel[c_i]$ for some $(c_j,c_\ell)\in\as(b_j)\tm\as(a_\ell)$,
a contradiction with the assumption that $\lnk_i$ is the equality relation.
Therefore $\relo$ is not linked, and, since $\zA_2,\zA_3$ are simple, 
$\relo$ is the graph of a bijection. Thus, $\zA_j$ and $\zA_\ell$ are 
isomorphic and there is an isomorphism that maps $a_j$ to $b_\ell$ and 
$b_j$ to $a_\ell$. In a similar way $\zA_i$ and $\zA_j$ are isomorphic. 

Next, we prove that for any $i\in[3]$, $\as(a_i)$ and $\as(b_i)$ are subalgebras
of $\zA_i$. Let $\zB_i=\Sg{\as(a_i)}$ and $\zC_i=\Sg{\as(b_i)}$. Consider 
$\rel''=\rel\cap(\zB_1\tm\zB_2\tm\zC_3)$. Since $(a_1,a_2,b_3)\in\rel$, by 
Corollary~\ref{cor:product-path} $\rel''$ is a subdirect product, and by Claim~2 
$\zB_2\tm\zC_3\sse\pr_{23}\rel''$. Let $\lnk''_1,\lnk''_{23}$ denote the link 
congruences  of $\zB_1,\zB_2\tm\zC_3$ when $\rel''$ is viewed as a subdirect 
product of $\zB_1$ and $\zB_2\tm\zC_3$. Clearly, $\lnk''_1\sse\lnk_1$, and so, 
$\lnk''_1$ is the equality relation, and therefore $\lnk''_{23}$ is a proper congruence 
of $\zB_2\tm\zC_3$. Moreover, $\rel''[a_1]$ is a block of this congruence that,
as we proved above is the graph of an isomorphism $\vf$ from $\zB_2$ to $\zC_3$.
Replacing $\zC_3$ with a copy of $\zB_2$ and $\vf$ with the identity
mapping, we see that $\zB_2^2$ has a congruence one of whose blocks
is $\Dl=\{(c,c)\mid c\in\zB_2\}$. Therefore $\zB_2$ is Abelian and, as $\zB_2$ 
omits type \one, by Theorem~9.6 of \cite{Hobby88:structure} 
and the results of~\cite{Herrmann79:affine} $\zB_2$ is term equivalent 
to a module. This 
means that for any $c\in\zB_2$ the pair $a_2c$ is a thin affine edge implying
$c\in\as(a_2)$. A similar argument shows that $\as(a_i)$ and $\as(b_j)$
are isomorphic modules for any $i,j\in[3]$. 

If, say, $a_1\in\as(b_1)$ then $\zA_1$ is a module and the result follows from 
the base case. If $a_1\not\in\as(b_1)$, then, as $a_1$ is as-maximal, by 
Theorem~\ref{the:connectivity} there is a special thin majority edge $dc$ such that 
$d\in\as(b_1)$ and $c\not\in\as(b_1)$. We show that this imples that $\lnk_1$
is the full relation contradicting the results above. Let $\zB=\Sg{d,c}$, and $\th$ 
a congruence witnessing that $dc$ a majority edge; in particular, 
$\zB=d\fac\th\cup c\fac\th$. The case $\zB=\zA_1$ is considered in Case~1, so, 
assume $\zB\ne\zA_1$. As in Case~1 consider 
$\rel'=\rel\cap(\zB\tm\zA_2\tm\zA_3)$ and take any $e\in\amax(\zB)\cap c\fac\th$.  
It was proved in the previous paragraph that $\rel\cap(\zC_1\tm\zB_2\tm\zB_3)$ 
is subdirect, therefore $(d,a'_2,a'_3)\in\rel$ for some 
$(a'_2,a'_3)\in\as(a_2)\tm\as(a_3)$. Then as before by Claim~2 
$(e,a'_i)\in\pr_{1i}\rel'$ for $i\in\{2,3\}$, and $(a'_2,a'_3)\in\pr_{23}\rel'$. 
By the induction hypothesis there is $(e',a''_2,a''_3)\in\rel'$ with $e'\in\as_\zB(e)$ 
(and so $e'\in c\fac\th$) and $a''_i\in\as_{\zA_i}(a_i)$, $i\in\{2,3\}$. 
Note that as $\as(b_1)$ is a module, $e'\not\in\as(b_1)$. On the other hand,
as $\rel\cap(\zC_1\tm\zB_2\tm\zB_3)$ is subdirect, there exists 
$b'\in\as(b_1)$ with $(b',a''_2,a''_3)\in\rel$, implying that $\lnk_1$ is not the
equality relation.
\end{proof}

%%%%%%%%%%%%%%%%%%%%%%%%%%%%%%%%%
\subsection{Proof of Theorem~\ref{the:quasi-2-decomp}}

In this section we prove Theorem~\ref{the:quasi-2-decomp}.

\begin{proof}[Proof of Theorem~\ref{the:quasi-2-decomp}]
Let $\ba$ be a tuple satisfying the conditions of 
quasi-2-decom\-po\-sa\-bi\-li\-ty  and such that $\pr_Y\ba\in\pr_Y\rel$, 
for $Y\in J$, and $\pr_X\ba\in\amax(\pr_X\rel)$. By induction on ideals of the 
power set of $[n]$ (i.e.\ subsets of the power set closed under 
taking subsets) we prove that for any ideal $I$ there is $\ba'$ 
such that $\pr_Y\ba\sqq^{as}\ba'$ for any $Y\in J$, 
%% $(\ba[i],\ba[j])\sqq^{as}(\ba'[i],\ba'[j])$ for any $i,j\in[n]$, 
$\pr_X\ba'\in\as(\pr_X\ba)$ (in particular, $\pr_X\ba\sqq^{as}\pr_X\ba'$), and for any 
$U\in I$ it holds $\pr_U\ba'\in\pr_U\rel$. Note that we cannot claim here that $\pr_U\ba\sqq^{as}\pr_U\ba'$, because $\pr_U\ba\not\in\pr_U\rel$, and it is not even clear what relation such a path may belong to. Instead we introduce a relation $\cE(I)$, specific for each $I$ that contains such paths, but only for sets $U\in I$. Then if this 
statement is proved for the entire power set, 
$\ba'=\pr_{[n]}\ba'\in\pr_{[n]}\rel=\rel$ implies the result.
The base case, where the ideal consists of all sets from $J$ and their 
subsets (this includes all at most 2-element sets and the set $X$ and its subsets) 
is given by the assumptions of the theorem and tuple $\ba$. 

Suppose that the claim is true for an ideal $I$, the set $W$ does not
belong to $I$, but all its proper subsets do. Let $\cE(I)$ be the set of
all tuples $\bc$, not necessarily from $\rel$ such that 
$\pr_U\bc\in\pr_U\rel$ for every $U\in I$. Clearly, $\rel\sse\cE(I)$
and $\cE$ is a subdirect product of $\tms\zA n$. By $\cD(I)$ we denote
the set of all tuples $\bc\in\cE(I)$ such that $\pr_Y\ba\sqq^{as}\pr_Y\bc$, 
for $Y\in J$.
It does not have to be a subalgebra of 
$\tms\zA n$. If a tuple belongs to $\cD(I)$ it is said to 
\emph{support} $I$. We show that $\cD(I)$ contains a tuple $\bb$ 
such that $\pr_W\bb\in\pr_W\rel$, that is, $\bb$ supports $I\cup\{W\}$. 

Assume that $W=[\ell]$. For a subalgebra $\relo$ of $\pr_W\rel$ a 
tuple $\bc\in\cD(I)$ is said to be \emph{$\relo$-approximable}
if $\pr_U\bc\in\pr_U\relo$. If $\bc$ is $\relo$-approximable, for every $U\subset W$ we choose $\bc_U\in\rel$ such that $\pr_U\bc=\pr_U\bc_U$.
%% for any $U\subset W$ there is $\bc_U\in\rel$ with 
%% $\pr_U\bc_U=\pr_U\bc$ and $\pr_W\bc_U\in\relo$. 
We prove the following statement:
\begin{quote} 
Let $\relo$ be a subalgebra of $\pr_W\rel$. If there exists 
a $\relo$-approximable tuple from $\cD(I)$, then there is 
$\bd\in\cD(I)$ such that $\pr_W\bd\in\relo$. 
\end{quote}

Note that if $\relo=\pr_W\rel$, then any tuple in $\cD(I)$ 
is $\relo$-approximable. Therefore the statement 
implies that $\cD(I)$ contains a tuple $\bd$ with 
$\pr_W\bd\in\pr_W\rel$, which would prove the induction step.
We prove the statement by induction on the sum of sizes of unary
projections of $\relo$. If one of these projections is 1-element then
the statement trivially follows from the assumption that for some 
$\bd\in\cD(I)$ it holds that $\pr_U\bd\in\pr_U\relo$
for $U$ including all coordinate positions whose projections contain
more than 1 element, and so $\pr_W\bd\in\pr_W\relo$. So suppose that 
the statement is proved for all relations with unary projections 
smaller than $\relo$.  Also let $\bd$ be a $\relo$-approximable tuple 
from $\cD(I)$. We need the following auxiliary statements.

\smallskip

{\sc Claim 1.}
Let $\bd\in\cD(I)$ and $\vc\be k$ an as-path in $\cE(I)$. Then every 
tuple from an as-path $\vc{\be'}k$ satisfying the following conditions belongs to 
$\cD(I)$: $\be'_1=\bd$; if $\be_i\be_{i+1}$ is a semilattice edge
then $\be'_{i+1}=\be'_i\cdot\be_{i+1}$ (in which case either $\be'_i=\be'_{i+1}$ or $\be'_i\be'_{i+1}$ is a semilattice edge by Proposition~\ref{pro:good-operation}); and if $\be_i\be_{i+1}$
is a thin affine edge then $\be'_i\be'_{i+1}$ is some thin affine edge
in $\Sgg{\tms\zA n}{\be'_i,\be_{i+1}}$. 

\smallskip

Since $\bd,\vc\be k\in\cE(I)$, we also have $\vc{\be'}k\in\cE(I)$. Thus, 
we only need to check that $\pr_Y\ba\sqq^{as}\pr_Y\be'_s$ 
for any $s\in[k]$ and $Y\in J$. However, this condition follows from 
the assumption $\bd\in\cD(I)$ --- thus, 
$\pr_Y\ba\sqq^{as}\pr_Y\be'_1$ for $Y\in J$, --- and that 
$\pr_Y\be'_1\zd\pr_Y\be'_k$ is an as-path for each $Y\in J$ by
Corollary~\ref{cor:product-path}. 

\smallskip

{\sc Claim 2.}
Let $\bc\in\cD(I)$ be $\relo$-approximable, $U\subset W$, and let 
$\be\in\pr_U\relo$ be such that $\pr_U\bc\sqq^{as}\be$ in $\pr_U\relo$.
Then there is $\bc'\in\cD(I)$ such that it is $\relo$-approximable,  
$\pr_U\bc'=\be$, and $\bc\sqq^{as}\bc'$ in $\cE(I)$.

\smallskip

Let $\pr_U\bc=\bb_1\zd\bb_k=\be$ be an as-path in $\pr_U\relo$.
Since $\bb_i\in\pr_U\relo$ for each 
$i\in[k]$, by Corollary~\ref{cor:product-path}(1) this path can 
be extended to an as-path $\pr_W\bc_U=\vc{\bb'}k$ in $\relo$. 
Then, since $\bb'_i\in\pr_W\rel$
for each $i\in[k]$, applying again Corollary~\ref{cor:product-path}(1) 
the as-path $\vc{\bb'}k$ can be extended to an as-path 
$\vc{\bb''}k$ in $\rel$ such that $\pr_W\bb''_i=\bb'_i\in\relo$, 
$\pr_U\bb''_i=\bb_i$ for each $i\in[k]$. We define a sequence 
$\vc{\bd}k$ as follows: $\bd_1=\bc$, if $\bb''_i\bb''_{i+1}$ is a 
thin semilattice edge then set $\bd_{i+1}=\bd_i\cdot\bb''_{i+1}$, 
and if $\bb''_i\bb''_{i+1}$ is a thin affine edge then choose 
$\bd_{i+1}$ such that $\bd_i\bd_{i+1}$ is a thin affine edge in 
$\Sgg{\tms\zA n}{\bd_i,\bb''_{i+1}}$ and $\pr_U\bd_{i+1}=\bb_{i+1}$. 
Note that, since $\pr_U\bd_i=\bb_i$ and 
$\pr_U\bb''_{i+1}=\bb_{i+1}$ such a $\bd_{i+1}$ exists by 
Lemma~\ref{lem:product-edge}(1). 
Now, set $\bc'=\bd_k$. By Claim~1 $\bc'$ belongs to $\cD(I)$.

To show that $\bc'$ is $\relo$-approximable, 
for any $V\subset W$ 
%% we construct a tuple $\bc'_V$ as follows.
consider the sequence 
$\pr_V\bc=\pr_V\bc_V=\be_1\zd\be_k=\pr_V\bd_k=\pr_V\be$,
where $\be_i=\pr_V\bd_i$. By construction this is an as-path in 
$\pr_V\relo$, so by Corollary~\ref{cor:product-path}(1) it can be 
extended to an as-path $\pr_W\bc_V=\be'_1\zd\be'_k$ in $\relo$ witnessing that $\pr_V\bc'\in\pr_V\relo$. 
%% Then, as above, this as-path can be extended to an as-path $\bc_V=\vc{\be''}k$ in $\rel$. The tuple $\be''_k$ can be chosen to serve as $\bc'_V$.

\smallskip

In particular, Claim~2 implies that a $\relo$-approximable tuple $\bc\in\cD(I)$ 
can be chosen such that 
for any $i\in W$ the element $\bc[i]$ is as-maximal in $\pr_i\relo$.
We will assume it from now on. For $U=W-\{i\}$, $i\in W$, we 
denote $\bc_U$ by $\bc_i$.
Suppose that for some $i\in W$ the unary projection
$\pr_i\relo\ne\Sg{\bc[i],\bc_i[i]}$. Assume $i=1$. Then set  
\[
\relo'=\relo\cap\left(\Sg{\bc[1],\bc_1[1]}\tm
\prod_{i\in  W-\{1\}}\pr_i\relo\right). 
\]
Note that in this case $\bc$ is $\relo'$-approximable, and the 
result follows by the inductive hypothesis.
Let $\bc_i$, $i\in W$, be chosen such that $\Sg{\bc[i],\bc_i[i]}=\pr_i\relo$. 
It is also clear that $\relo$ can be chosen to be $\Sg{\vc{\pr_W\bc}\ell}$.

\smallskip

{\sc Claim 3.}
For any $i\in W$ there is $\bd\in\relo$ such that $\bc[i]\sqq^{as}\bd[i]$
and $\pr_U\bc\sqq^{as}\pr_U\bd$ in $\pr_U\relo$ where $U=W-\{i\}$. 

\smallskip

Without loss of generality assume $i=1$ and $U=W-\{i\}$. Consider 
the relation
\[
\relo'(x,y,z)=\exists x_4\zd x_\ell (\relo(x,y,z,x_4\zd
x_\ell)\meet(x_4=\bc[4]) \meet\ldots\meet(x_\ell=\bc[\ell])).
\]
Obviously, $\pr_{\{1,2,3\}}\bc_1,\pr_{\{1,2,3\}}\bc_2,
\pr_{\{1,2,3\}}\bc_3\in\relo'$. We show that $\relo'$ contains a tuple 
$\bd$ such that $\pr_{12}\bc\sqq^{as}\pr_{12}\bd$, 
$\pr_{13}\bc\sqq^{as}\pr_{13}\bd$, $\pr_{23}\bc\sqq^{as}\pr_{23}\bd$.
This would imply the claim, because, $\pr_{23}\bc\sqq^{as}\pr_{23}\bd$
means $\pr_U\bc\sqq^{as}\pr_U\bd$, and any of the first two connections
means that $\bc[i]\sqq^{as}\bd[i]$. 

If, say, the pair $(\bc[1],\bc[2])$ is not as-maximal in $\pr_{12}\relo'$, choose
an as-path $(\bc[1],\bc[2])=\be_1\zd \be_s$ in $\pr_{12}\relo'$ 
such that $\be_s$ is as-maximal in $\pr_{12}\relo'$. By 
Corollary~\ref{cor:product-path}(1) this as-path can be extended 
to an as-path $\pr_{\{1,2,3\}}\bc_3=\be'_1\zd\be'_s$ in $\relo'$. 
Now, as in the proof of Claim~2 we construct a sequence
$\pr_{\{1,2,3\}}\bc=\be''_1\zd\be''_s$, that is not necessarily from
$\relo'$, as follows. If $\be'_i\le\be'_{i+1}$, set 
$\be''_{i+1}=\be''_i\cdot\be'_{i+1}$. If $\be'_i\be'_{i+1}$ is
a thin affine edge, then set $\be''_{i+1}$ to be any tuple in
$\Sg{\be''_i,\be'_{i+1}}$ such that $\be''_i\be''_{i+1}$ is a thin
affine edge and $\pr_{12}\be''_{i+1}=\pr_{12}\be'_{i+1}$. 
%% Such a tuple exists by Lemma~\ref{lem:product-edge}. 
The resulting tuple $\be''_s$ satisfies the following conditions: 
$\pr_{12}\be''_s=\be_s$
and 
\[
\pr_{13}\bc\sqq^{as}\pr_{13}\be''_s\in\pr_{13}\relo',\qquad
\pr_{23}\bc\sqq^{as}\pr_{23}\be''_s\in\pr_{23}\relo'.
\]

Repeating the procedure above for projections on $\{1,3\}$ and 
$\{2,3\}$ if necessary, we obtain a tuple $\bc'$ such that $\pr_{ij}\bc'$
is an as-maximal tuple in $\pr_{ij}\relo'$ for $i,j\in[3]$. Relation 
$\relo'$ and the tuples $\pr_{ij}\bc'$, $i,j\in[3]$, satisfy
the conditions of Lemma~\ref{lem:3-quasi}. Therefore, $\relo'$
contains a tuple $\bd$ such that $\pr_{ij}\bc\sqq^{as}\pr_{ij}\bd$ 
in $\pr_{ij}\relo'$ for $i,j\in[3]$. The result follows.

\smallskip

To complete the proof let $U=[\ell-1]$ (recall that $W=[\ell]$) and 
$\bd$ the tuple obtained in Claim~3 for $i=\ell$.
Then $\pr_U\bc\sqq^{as}\pr_U\bd$ in $\pr_U\relo$. 
By Claim~2 $\bc$ can be amended so that the new tuple $\bc'$ 
still supports $I$, but $\pr_U\bc'=\pr_U\bd$. Note that $\bc'[\ell]$ is 
in the same as-component of $\pr_\ell\relo$ as $\bd[\ell]$, therefore, 
$\bc'[\ell]\sqq^{as}\bd[\ell]$ in $\pr_\ell\relo$. 
If $\Sgg{\zA_\ell}{\bc'[\ell],\bd[\ell]}=\pr_\ell\relo$, then there is an as-path
from $\pr_W\bc'$ to $\bd$ in $\pr_W\cE(I)$: this is because 
$\pr_U\bc'=\pr_U\bd$ and $\{\pr_U\bd\}\tm\pr_\ell\relo\sse\pr_W\cE(I)$ 
in this case. Extend this path to a path in $\cE(I)$ as before, and 
change $\bc'$ using this path as in Claim~1. The resulting tuple 
$\bc''$ supports $I$, and 
$\pr_W\bc''=\bd\in\relo$ as required.
%% , implying $\bc''$ supports $I\cup\{W\}$.
If $\zB=\Sgg{\zA_\ell}{\bc'[\ell],\bd[\ell]}\ne\pr_\ell\relo$, then set 
$\relo''=\relo\cap(\pr_{[\ell-1]}\relo\tm\zB)$. Since $\bc'$
is $\relo''$-approximable, the result follows from the inductive
hypothesis.

To finish the proof of Theorem~\ref{the:quasi-2-decomp} it suffices to 
take care of the requirement that the resulting tuple $\bb$ is such that 
$\pr_X\bb=\pr_X\ba$. Since $X\in I$ already in the base
case, the resulting tuple $\bb$ is such that $\pr_X\ba\sqq^{as}\pr_X\bb$. By 
Corollary~\ref{cor:product-path}(1) there is also a tuple $\bb'$ 
satisfying the same requirements and such that $\pr_X\bb'=\pr_X\ba$. 
\end{proof}

Now we generalize Theorem~\ref{the:quasi-2-decomp} to the variety
$\cV$ generated by $\cK$.

\begin{corollary}\label{cor:quasi-2-decomp-variety}
Let $\vc\zA n\in\cV$. Then any subdirect product $\rel$ of
$\tms\zA n$ is quasi-2-decomposable.
\end{corollary}

\begin{proof}
We first prove that quasi-2-decomposability holds when every $\zA_i$ 
is a subalgebra of a direct product of algebras from $\cK$. Let
$\zA_i$ be a subalgebra of $\zA_{i1}\tm\dots\tm\zA_{1r_i}$, where 
$\zA_{i1}\zd\zA_{ir_i}\in\cK$; we may assume $\zA_i$ is subdirect. 
Let also $\rel$ be a subdirect product of $\vc\zA n$ and $\ba\in\tms\zA n$
satisfying the premise of quasi-2-decomposability. 
Then $\rel$ can be represented as a subdirect product of 
\[
\prod_{i=1}^n\prod_{j=1}^{r_i}\zA_{ij},
\]
and $\ba$ as a tuple from this product. Let $\rel',\ba'$ denote such 
representations. We then apply Theorem~\ref{the:quasi-2-decomp}
to $\rel'$, $\ba'$, and a collection $J$ that is derived from 2-element
subsets of $[n]$, that is, 
\[
J=\{\{(i,1)\zd(i,r_i),(j,1)\zd(j,r_j)\}\mid i,j\in[n]\}.
\]

Next, let $\zA_i=\zA'_i\fac{\th_i}$, $\th_i\in\Con(\zA'_i)$, and 
quasi-2-decomposability hold for any subdirect product of the $\zA'_i$.
For a subdirect product $\rel\sse\tms\zA n$ and $\ba\in\tms\zA n$ 
satisfying the premise of quasi-2-decomposability, let 
\[
\rel'=\{(\vc bn)\in\tms{\zA'}n\mid (b_1\fac{\th_1}\zd b_n\fac{\th_n})\in\rel\}
\]
and $\ba'\in\tms{\zA'}n$ be such that $\ba[i]=\ba'[i]\fac{\th_i}$, $i\in[n]$. 
Then by Theorem~\ref{the:quasi-2-decomp} $\rel'$ is quasi-2-decomposable and
there is $\bb'\in\rel'$ such that 
$(\ba'[i],\ba'[j])\sqq^{as}_{\pr_{ij}\rel'}(\bb'[i],\bb'[j])$. Then 
$\bb=(\bb'[1]\fac{\th_1}\zd\bb'[n]\fac{\th_n})\in\rel$ and by 
Corollary~\ref{cor:quotient-path} satisfies the requirements of
quasi-2-decomposability.
\end{proof}

Another consequence of Theorem~\ref{the:quasi-2-decomp} is the 
existence of a very useful term.

\begin{theorem}\label{the:pseudo-majority}
%% Let $\cK$ be a finite class of finite similar smooth algebras
%% omitting type~\one\ and $\cV$ the variety it generates. 
There is a term operation $\maj$ of $\cV$ such that for any $\zA\in\cV$ 
and any $a,b\in\zA$, $\maj(a,a,b),\maj(a,b,a),\maj(b,a,a)\in\Filt_\zA^{as}(a)$.

In particular, if $a$ is as-maximal, the elements $\maj(a,a,b)$, $\maj(a,b,a)$, $\maj(b,a,a)$ 
belong to the as-component of $\zA$ containing $a$.
\end{theorem}

\begin{proof}
Since $\cV$ is finitely generated, it suffices to prove the result for finite 
subsets of $\cV$ and then apply the compactness argument. So, let
$\cK'\sse\cV$ be a finite set of finite algebras.

Let $\{a_1,b_1\}\zd\{a_n,b_n\}$ be a list of all pairs of elements from 
algebras of $\cK'$, let $a_i,b_i\in\zA_i$. Define a relation $\rel$ to be 
the subalgebra of the product of $\zA_1^3\tm\dots\tm\zA_n^3$ 
generated by $\ba_1,\ba_2,\ba_3$, where for every 
$i\in[n]$, $\pr_{3i-2,3i-1,3i}\ba_1=(a_i,a_i,b_i)$, 
$\pr_{3i-2,3i-1,3i}\ba_2=(a_i,b_i,a_i)$, $\pr_{3i-2,3i-1,3i}\ba_3=(b_i,a_i,a_i)$.
In other words the triples $(\ba_1[3i-2],\ba_2[3i-2],\ba_3[3i-2])$,
$(\ba_1[3i-1],\ba_2[3i-1],\ba_3[3i-1])$, $(\ba_1[3i],\ba_2[3i],\ba_3[3i])$ 
have the form $(a_i,a_i,b_i),(a_i,b_i,a_i),(b_i,a_i,a_i)$, respectively. Therefore 
it suffices to show that $\rel$ contains a tuple $\bb$ such that 
$a_i\sqq^{as}\bb[j]$, where $j\in\{3i,3i-1,3i-2\}$. However, since 
$(a_{i_1},a_{i_2})\in\pr_{j_1j_2}\rel$ for any $i_1,i_2\in[n]$ and
$j_1\in\{3i_1,3i_1-1,3i_1-2\}$, $j_2\in\{3i_2,3i_2-1,3i_2-2\}$, this 
follows from Corollary~\ref{cor:quasi-2-decomp-variety}.
\end{proof}

A function $\maj$ satisfying the properties from 
Theorem~\ref{the:pseudo-majority} will be called a 
\emph{quasi-majority function}. 

%%%%%%%%%%%%%%%%%%%%%%%%%%%%%%
%%%%%%%%%%%%%%%%%%%%%%%%%%%%%%
\section{Rectangularity for maximal components}%
\label{sec:max-rectangularity}

In this section we show a stronger rectangularity property --- involving 
multi-ary relations --- than that in Proposition~\ref{pro:max-gen}, 
but for maximal components, rather than as-components. 

An algebra $\zA$ is said to be \emph{maximal generated} if it is
generated by one of its maximal components.

%%%%%%%%%%%%%%%%%%%%%%%%%%%%%%
\subsection{Simple maximal generated algebras}

We start with several auxiliary statements.

\begin{lemma}\label{lem:3-ary}
Let $\rel$ be a subdirect product of simple maximal generated algebras
$\zA_1,\zA_2,\zA_3\in\cV$, generated by their maximal components 
$C_1,C_2,C_3$, respectively. If $\zA_i\tm \zA_j\sse\pr_{ij}\rel$ 
for every $i,j\in[3]$ and $\rel\cap(C_1\tm C_2\tm C_3)\ne\eps$, 
then $\rel=\zA_1\tm \zA_2\tm \zA_3$.
\end{lemma}

\begin{proof}
Note that if $|\zA_i|=1$ the statement is trivial, so we assume 
$|C_i|>1$ for $i\in[3]$. We argue as in the proof of 
Lemma~\ref{lem:3-quasi}. Consider $\rel$ as a subdirect product of 
$\zA_1$ and $\zA_2\tm\zA_3$. Since $\zA_1$ is simple, the link
congruence $\lnk_1$ is either the equality relation or the full congruence. 
In the latter case, as $\rel\cap(C_1\tm C_2\tm C_3)\ne\eps$, by 
Proposition~\ref{pro:max-gen} $C_1\tm C_2\tm C_3\sse\rel$ and the
result follows. Suppose that $\lnk_1$ is the equality relation.

Recall that for $a\in \zA_1$ by $\rel[a]$ we 
denote the set $\rel[a]=\{(b_2,b_3)\mid (a,b_2,b_3)\in \rel\}$. 
Notice that, for every $a\in\zA_1$, $\rel[a]$ is a subalgebra of
$\pr_{23}\rel$, and, since $\pr_{12}\rel=\zA_1\tm\zA_2$,
$\pr_{13}\rel=\zA_1\tm\zA_3$, the algebra $\rel[a]$ is a subdirect
product of $\zA_2,\zA_3$. Since both $\zA_2,\zA_3$ are simple, $\rel[a]$ 
is either linked or the graph of a bijective mapping. Let
$(a_1,a_2,a_3)\in\rel\cap(C_1\tm C_2\tm C_3)$. If $\rel[a_1]$ is linked,
by Proposition~\ref{pro:max-gen} $C_2\tm C_3\sse\rel[a_1]$, which 
contradicts the assumption that $\lnk_1$ is the equality relation. Therefore
$\rel[a_1]$ is the graph of a mapping $\vf:\zA_2\to\zA_3$. The mapping
$\vf$ is an isomorphism between $\zA_2$ and $\zA_3$. Replacing 
$\zA_3$ with an isomorphic copy of $\zA_2$, the link congruence $\lnk_{23}$
(when considering $\rel$ as a subdirect product of $\zA_1$ and 
$\zA_2\tm\zA_3$) is a nontrivial congruence of $\zA_2^2$ one of whose 
blocks is $\Dl=\{(a,a)\mid a\in\zA_2\}$. This implies that $\zA_2$ is a 
module. In particular, $|C_2|=1$, a contradiction.
\end{proof}

Observe that maximal components in Lemma~\ref{lem:3-ary} (and 
hence in the remaining results from this section) cannot be replaced 
with as-components. Indeed, that would include subdirect products of 
modules, for which Lemma~\ref{lem:3-ary} is not true. 

The next lemma generalizes Lemma~\ref{lem:3-ary} to products of multiple algebras. 

\begin{lemma}\label{lem:n-ary}
Let $\rel$ be a subdirect product of simple maximal generated
algebras $\zA_1\zd\zA_n\in\cV$, say, $\zA_i$ is generated by a maximal 
component $C_i$. If $\zA_i\tm \zA_j\sse\pr_{ij}\rel$ for 
every $i,j\in[n]$ and $\rel\cap(\tms Cn)\ne\eps$, then 
$\rel=\zA_1\tm\ldots\tm \zA_n$.
\end{lemma}

\begin{proof}
We prove the lemma by induction. The base case of induction $n=3$
has been proved in Lemma~\ref{lem:3-ary}. Suppose that
the lemma holds for each number less than~$n$. Take $a\in
C_1$ and recall that $\rel[a]=\{(b_2\zd b_n)\mid
(a,b_2\zd b_n) \in \rel\}$.  By Lemma~\ref{lem:3-ary}, $\zA_1\tm
\zA_i\tm \zA_j\sse \pr_{1,i,j}\rel$ for any $2\le i,j\le n$.
Then $\zA_i\tm \zA_j\sse\pr_{ij}\rel[a]$. Therefore by induction 
hypothesis $\rel[a]= \zA_2\tm\ldots\tm \zA_n$.  The lemma then follows from the assumption that $C_1$ generates $\zA_1$.
\end{proof}

Lemma~\ref{lem:n-ary} allows one to describe the structure of 
subdirect products of simple maximal generated algebras.

\begin{defin}
A relation $\rel\sse \zA_1\tm\ldots\tm \zA_n$ is said to be {\em
almost trivial} if there exists an equivalence relation $\th$ on the set
$[n]$ with classes $I_1\zd I_k$, such that
\[
\rel=\pr_{I_1}\rel\tm\ldots\tm\pr_{I_k}\rel
\]
where $\pr_{I_j}\rel=\{(a_{i_1},\pi_{i_2}(a_{i_1})\zd
\pi_{i_l}(a_{i_1}))\mid a_{i_1}\in \zA_{i_1}\}$, $I_j=\{i_1\zd i_l\}$,
for certain bijective mappings $\pi_{i_2}\colon \zA_{i_1}\to
\zA_{i_2}\zd \pi_{i_l}\colon \zA_{i_1}\to \zA_{i_l}$.
\end{defin}

\begin{lemma}\label{lem:at-for-simple}
Let $\rel$ be a subdirect product of simple maximal generated  
algebras $\vc\zA n$, say, $\zA_i$ is generated by a maximal component 
$C_i$; and let $\rel\cap(C_1\tm\ldots\tm C_n)\ne\eps$. Then $\rel$ 
is an almost trivial relation.  
\end{lemma}

\begin{proof}
We prove the lemma by induction on $n$. When $n=1$ the result 
holds trivially. 

We now prove the induction step. By Proposition~\ref{pro:max-gen}, 
for any pair $i,j\in[n]$ the projection $\pr_{ij}\rel$ is either $\zA_i\tm\zA_j$, 
or the graph of a bijective mapping. Assume that there exist $i,j$
such that $\pr_{ij}\rel$ is the graph of a mapping
$\pi\colon\zA_i\to\zA_j$. By the inductive hypothesis $\pr_{[n]-\{j\}}\rel$ 
is almost trivial, and therefore can be represented in the form
\[
\pr_{[n]-\{j\}}\rel=\pr_{I_1}\rel\tm\ldots\tm\pr_{I_k}\rel
\]
where $I_1\cup\ldots\cup I_k=[n]-\{j\}$. Suppose, for simplicity,
that $i$ is the last coordinate position in $I_1$, that is,
\begin{eqnarray*}
\pr_{I_1}\rel &=& \{(a_{i_1}\zd a_{i_k},a_i)\mid a_{i_1}\in\zA_{i_1}, \
a_{i_s}=\pi_{s1}(a_{i_1})\\
& & \hbox{for}\ s\in\{2\zd k\},\ a_i=\pi_i(a_{i_1})\}.
\end{eqnarray*}
Then 
\begin{eqnarray*}
\pr_{I_1\cup\{j\}}\rel &=& \{(a_{i_1}\zd a_{i_k},a_i,a_j)\mid a_{i_1}\in\zA_{i_1}, \
a_{i_s}=\pi_{s1}(a_{i_1})\\
& & \hbox{for}\ s\in\{2\zd k\},\ a_i=\pi_i(a_{i_1}),\ a_j=\pi\pi_i(a_{i_1})\},
\end{eqnarray*}
and we have $\rel=\pr_{I_1\cup\{j\}}\rel\tm\ldots\tm\pr_{I_k}\rel$, as
required. 

Finally, if $\pr_{ij}\rel=\zA_i\tm\zA_j$ for all $i,j\in\un n$, then
the result follows by Lemma~\ref{lem:n-ary}.
\end{proof}

%%%%%%%%%%%%%%%%%%%%%%%%%%%%%%%%%%%%
\subsection{General maximal generated algebras}

Here we consider the case when factors of a subdirect product are maximal
generated, but not necessarily simple.

\begin{lemma}\label{lem:n-simple}
Suppose that $\rel$ is a subdirect product of maximal generated algebras
$\zA_1\zd\zA_n\in\cV$, where $\zA_1$ is simple. Let also $\zA_1$
be generated by a maximal component $C_1$, $\pr_{2\zd n}\rel$ 
is maximal generated, say, by a maximal component $\relo$, 
$\rel\cap(C_1\tm\relo)\ne\eps$, and $\pr_{1i}\rel=\zA_1\tm \zA_i$ for 
$i\in\{2\zd n\}$. Then $\rel=\zA_1\tm\pr_{2\zd n}\rel$.
\end{lemma}

\begin{proof}
We prove the lemma by induction on $n$. The case $n=2$ is obvious. 
Consider the case $n=3$. We use induction on 
$|\zA_1|+|\zA_2|+|\zA_3|$. The trivial case
$|\zA_1|+|\zA_2|+|\zA_3|=3$ gives the base  
case of induction. Let $\relo$ be a maximal component of $\pr_{23}\rel$
generating it. If both $\zA_2,\zA_3$ are simple, 
then the result follows from Lemma~\ref{lem:at-for-simple}. Otherwise, 
suppose that $\zA_3$ is not simple. 

Note that by Corollary~\ref{cor:product-path}(1) $\rel$ contains some subdirect product of $C_1$ and $\relo$. Take a maximal congruence $\th$
of $\zA_3$, fix a $\th$-class $D$ such that $\pr_3\relo\cap D\ne\eps$ (we keep the  indexing of coordinates as in $\rel$) and consider
$\rel\fac\th\sse\zA_1\tm\zA_2\tm\zA_3\fac\th$, $\rel_D\sse\rel$ such that
\begin{eqnarray*}
\rel\fac\th  &=& \{(a,b,c\fac\th)\mid (a,b,c)\in\rel\},\\
\rel_D &=& \{(a,b,c)\mid (a,b,c)\in\rel, c\in D\}.
\end{eqnarray*}
By what is observed above $\rel_D\cap(C_1\tm\relo)\ne\eps$. Pick 
$\ba\in\rel_D\cap(C_1\tm\relo)$ and consider a maximal component $\rela$ 
of $\rel_D$ such that $\ba\sqq\bb$ for some $\bb\in\rela$. By 
Corollary~\ref{cor:product-path}(1) $\rela\sse\rel_D\cap(C_1\tm\relo)$. Let 
$\rel'\sse\rel$ be the algebra generated by $\rela$. Clearly, 
$\pr_1\rela\cap C_1\ne\eps$ and 
$\pr_{23}\rela\cap\relo\ne\eps$. Note that since 
$\zA_1\tm\zA_3\sse\pr_{13}\rel$, by Corollary~\ref{cor:product-path}(1) 
this implies $\pr_1\rela\cap C_1=C_1$. Also, obviously,
$\pr_{13}\rel_D=\zA_1\tm D$. The projection $\pr_3\rela$ is a maximal 
component of $D$; denote it $E$ and denote by $E'$ the algebra generated by 
$E$. Since $\zA_1\tm E\sse\pr_{13}\rel$ and 
$(C_1\tm E)\cap\pr_{13}\rel'\ne\eps$, again by 
Corollary~\ref{cor:product-path}(1) we obtain $C_1\tm E\sse\pr_{13}\rel'$. Therefore $\zA_1\tm E'\sse\pr_{13}\rel'$.
By Proposition~\ref{pro:max-gen}, $\pr_{23}\rel\fac\th$ is either the graph of a 
mapping, or $\zA_2\tm\zA_3\fac\th$.  

\bigskip

\noindent
{\sc Case 1.}
$\pr_{23}\rel\fac\th$ is the graph of a mapping
$\pi\colon\zA_2\to\zA_3\fac\th$. 

\medskip

As before, Corollary~\ref{cor:product-path}(1) implies that $F=\pr_2\rela$ is a maximal 
component of $B=\pi^{-1}(D)$.  Then $F'=\pr_2\rel'$ is generated by $F$.
Since for each $(a,b)\in\zA_1\tm B\sse\pr_{12}\rel$ there is $c\in D$ with
$(a,b,c)\in\rel$, we have $\zA_1\tm B\sse\pr_{12}\rel_D$. Also, 
$\pr_{12}\rela\cap(C_1\tm F)\ne\eps$, say, $\ba\in\pr_{12}\rela\cap(C_1\tm F)$. 
As $\rela$ is a maximal component in $\rel_D$, for any $\bb\in C_1\tm F$ such 
that $\ba\sqq\bb$, we have $\bb\in\pr_{12}\rela\cap(C_1\tm F)$. This implies 
$C_1\tm F\sse\pr_{12}\rela$. Thus,  
$\pr_{12}\rel'=\zA_1\tm F'$. 

Since $|\zA_1|+|F'|+|E'|<|\zA_1|+|\zA_2|+|\zA_3|$, and
$\pr_{23}\rel'$ is maximal generated, inductive hypothesis
implies $\zA_1\tm\pr_{23}\rel'\sse\rel'$. In particular, 
there is $(a,b)\in\pr_{23}\rel'\cap\relo\sse\pr_{23}\rel$ such that
$\zA_1\tm\{(a,b)\}\sse\rel$. To finish the proof we just apply
Lemma~\ref{lem:buket}.

\bigskip

\noindent
{\sc Case 2.}
$\pr_{23}\rel\fac\th=\zA_2\tm\zA_3\fac\th$. 

\medskip

Since $|\zA_1|+|\zA_2|+|\zA_3\fac\th|<|\zA_1|+|\zA_2|+|\zA_3|$,
$\zA_3\fac\th$ is simple, and $\pr_{12}\rel=\zA_1\tm\zA_2$, by
inductive hypothesis, $\rel\fac\th=\zA_1\tm\zA_2\tm\zA_3\fac\th$. Therefore,
$\pr_{12}\rel_D=\zA_1\tm\zA_2$. Then $\pr_{12}\rel'=\zA_1\tm\zA_2$. 
Indeed, let $C_2=\pr_2\relo$, it is a maximal component of $\zA_2$, and
$\zA_2$ is generated by $C_2$. Moreover, $\zA_1\tm\zA_2$ is generated
by $C_1\tm C_2$. By the choice of $\rel'$, there is 
$(a,b,c)\in\rel'\cap(C_1\tm C_2\tm E)$. By Corollary~\ref{cor:product-path}(1)
for any $(a',b')\in C_1\tm C_2$ there is a path from $(a,b,c)$ to $(a',b',c')$ for
some $c'\in E$.

Now we argue as in Case~1, except in this case $F'=\zA_2$.

\medskip

Let 
us assume that the lemma is proved for $n-1$. Then
$\zA_1\tm\pr_{3\zd n}\rel\sse\pr_{1,3\zd n}\rel$. Denoting $\pr_{3\zd
n}\rel$ by $\rel'$ we have $\rel\sse\zA_1\tm\zA_2\tm\rel'$, and the
conditions of the lemma hold for this subdirect product. Thus
$\rel=\zA_1\tm\pr_{2\zd n}\rel$ as required.
\end{proof}

Lemma~\ref{lem:n-simple} serves as the base case for the following  more
general statement.

\begin{lemma}\label{lem:relation-direct-product}
Let $\rel$ be a subdirect product of algebras $\zA_1\zd\zA_n\in\cV$. 
Let $\zA_1$ be generated by a maximal component $C_1$, $\pr_{2\zd n}\rel$ 
is maximal generated, say, by a maximal component $\relo$, 
$\rel\cap(C_1\tm\relo)\ne\eps$, and $\pr_{1,i}\rel=\zA_1\tm \zA_i$ for 
$i\in\{2\zd n\}$, Then $\rel=\zA_1\tm\pr_{2\zd n}\rel$.
\end{lemma}

\begin{proof}
For every $i\in\{2\zd n\}$ the set $C_i=\pr_i\relo$ is a maximal component.
Moreover, $\zA_i$ is generated by $C_i$, and therefore is also maximal
generated.
 
We show that for any subalgebra $\rela$ of $\zA_1$ such that
(a) $\rela\cap C_1\ne\eps$, (b) $\rela$ is maximal generated by its
elements from $C_1$, 
%% (c) $\max(C_1\cap\rela)\tm C_i\sse\pr_{1i}\rel$ for $i\in\{2\zd n\}$, 
and 
(c) $\relo\sse\pr_{2\zd n}(\rel\cap(\rela\tm C_2\tm\dots\tm C_n))$,
the following holds: $\{d\}\tm\relo\sse\rel$ for any $d\in\max(C_1\cap\rela)$.

We prove by induction on the size of $\rela$. If $|\rela|=1$, then its
only element belongs to $C_1$ by (a) and is maximal. Then (c) is
equivalent to the claim. Suppose that the result holds for all
subalgebras satisfying (a)--(c) smaller than $\rela$. If
$\rela$ is simple then the result follows from Lemma~\ref{lem:n-simple},
since conditions (a)--(c) imply the premises of 
Lemma~\ref{lem:n-simple}. (Note that $\pr_{1,i}\rel=\zA_1\tm \zA_i$ implies $\max(C_1\cap\rela)\tm C_i\sse\pr_{1i}\rel$ for $i\in\{2\zd n\}$.) Otherwise let $\th$ be a maximal congruence of
$\rela$, let $\rel'=\{(c_1,c_2\zd c_n)\in\rel\mid c_1\in\rela\}$, and let  
\[
\rel^\th=\{(c_1\fac\th,c_2\zd c_n)\mid (c_1,c_2\zd c_n)\in \rel'\}.
\]
By Lemma~\ref{lem:n-simple}
$\rel^\th=\rela\fac\th\tm\pr_{2\zd n}\rel$. Take a class $\rela'$ of
$\th$ containing elements from $C_1$. Observe that $\rela'$
satisfies condition (c). As is easily seen there is a maximal component $B$
of $\rela'$ containing elements from $C_1$. Indeed, take any
$d\in C_1\cap\rela'$, then $\Filt_{S'}^s(d)\sse C_1$. Let $\rela''$
be a subalgebra of $\rela'$ generated by $B$, we show it satisfies
(a)--(c). 

Conditions (a) and (b) are true by the choice of $\rela''$.
%% , condition (c) holds since this condition is true for $\rela$. 
For condition (c) observe first that
$\relo\sse\pr_{2\zd n}(\rel\cap(\max(\rela')\tm \relo))$.
As for any $d\in C_1\cap \rela'$ there is $(a_2\zd a_n)\in\relo$ 
with $(d,a_2\zd a_n)\in\rel$ (it follows from the condition 
$\rel\cap(C_1\tm\relo)\ne\eps$ of the lemma and 
Corollary~\ref{cor:product-path}(1)), applying 
Corollary~\ref{cor:product-path}(1)  again
$\relo\sse \pr_{2\zd n}(\rel\cap(B\tm\relo))$,  
and (c) is also true for $\rela''$. By the inductive hypothesis 
$\{d\}\tm\relo\sse \rel$ for $d\in B$. Applying 
Proposition~\ref{pro:max-gen} we obtain the result.  

Finally, since $\zA_1$ contains $C_1$ and satisfies conditions (a)--(c), it follows 
that $C_1\tm\relo\sse\rel$, and the result is proved. 
\end{proof}

\begin{corollary}\label{cor:max-comp-product}
Let $\rel$ be a subdirect product of algebras
$\zA_1\zd\zA_n\in\cV$ such that $\pr_{1i}\rel$ is linked for any $i\in\{2\zd n\}$. 
Let also $\ba\in\rel$ be such that $\ba[1]\in\max(\zA_1)$ and 
$\pr_{2\dots n}\ba\in\max(\pr_{2\dots n}\rel)$. Then 
$\se{\ba[1]}\tm\se{\pr_{2\dots n}\ba}\sse\rel$.
\end{corollary}

\begin{proof}
Consider $\rel'$, the relation generated by 
$\rel\cap(C\tm D)$, $C=\se{\ba[1]}$, $D=\se{\pr_{2\dots n}\ba}$. 
By Corollary~\ref{cor:product-path} $\se{\ba[1]}\sse\pr_1\rel'$ and
$\se{\pr_{2\dots n}\ba}\sse\pr_{2\dots n}\rel'$. Moreover, $C$ and 
$D$ are maximal components in $\pr_1\rel'$ and $\pr_{2\dots n}\rel'$,
respectively. Since $\pr_{1i}\rel$ is linked for $i\in\{2\zd n\}$, 
$C\tm\pr_i D\sse\pr_{1i}\rel$ by Proposition~\ref{pro:max-gen}. 
As $\rel'$ is generated by $\rel\cap(C\tm D)$, $\pr_{1i}\rel'$ is generated 
by $\pr_{1i}\rel\cap(C\tm\pr_i D)$ implying $C\tm\pr_i D\sse\pr_{1i}\rel'$.
Subalgebras $\pr_1\rel'$ and $\pr_i\rel'$ are generated by $C$ and
$\pr_i D$, respectively, therefore, $\pr_{1i}\rel'=\pr_1\rel'\tm\pr_i\rel'$.
Now by Lemma~\ref{lem:relation-direct-product} the result follows.
\end{proof}

%%%%%%%%%%%%%%%%%%%%%%%%%%%%%%%%%%
%%%%%%%%%%%%%%%%%%%%%%%%%%%%%%%%%%
\section*{Declarations}

\subsection*{Data availability}
Data sharing not applicable to this article as datasets were neither generated nor analyzed.

\subsection*{Compliance with ethical standards}
The author is a member of the Editorial Board of Algebra Universalis. Apart from this the author declares that he has no conflict of interest.

\bibliographystyle{spmpsci}
%% \bibliography{one}

\end{document}